\UseRawInputEncoding
\documentclass[preprint,11pt]{article}

\newcommand\numberthis{\addtocounter{equation}{1}\tag{\theequation}}

\usepackage{array, makecell, cellspace}

\usepackage{setspace}

\usepackage{adjustbox}

\newcommand{\fdfrac}[2]{\mbox{\footnotesize$\displaystyle\frac{#1}{#2}$}}

\makeatletter
\newcommand{\vast}{\bBigg@{3.5}}
\makeatother

\usepackage{floatpag}
\usepackage{afterpage}

\usepackage[bottom]{footmisc}

\usepackage{enumitem}

\usepackage{mdframed}

\usepackage{float}

\usepackage[linesnumbered,lined]{algorithm2e}

\SetCommentSty{mycommfont}
\DontPrintSemicolon

\usepackage[T1]{fontenc}
\usepackage{lmodern}
\usepackage[utf8]{inputenc}

\usepackage{footnote}
\usepackage{tablefootnote}
\makesavenoteenv{tabular}

\usepackage{newfloat}
\DeclareFloatingEnvironment[fileext=lop]{Algorithm}

\usepackage[labelfont=bf]{caption}

\usepackage{geometry}
\usepackage{framed}
\usepackage{float}
\restylefloat{figure}
\usepackage{lineno}
\usepackage{amssymb}
\usepackage{amsmath}
\usepackage{amsthm}
\usepackage{dsfont}
\usepackage{bbm}
\usepackage{mathtools}
\usepackage{hyperref}
\newtheorem{definition}{Definition}[section]
\newtheorem{theorem}{Theorem}
\newtheorem{corollary}{Corollary}[theorem]
\newtheorem{lemma}[theorem]{Lemma}

\usepackage[table]{xcolor}
\usepackage{array, makecell}
\usepackage[normalem]{ulem} 
\usepackage{graphicx}
\usepackage{relsize}
\usepackage[numbers]{natbib}

\newcommand{\veps}{\varepsilon}

\newcommand{\X}{\mathcal{X}}

\newcommand{\st}[1][m]{ \sqrt{ #1 \left( #1 - 1 \right) } }
\newcommand{\del}[1][F]{\left( \sqrt{2}+1 \right) \delta(#1)}

\newcommand{\vc}[1][Vertex Cover ]{{\fontfamily{cmss}\selectfont
#1}}

\allowdisplaybreaks

\newtheorem{fact}{Fact}
\usepackage{xcolor}
\usepackage{thmtools,thm-restate}

\definecolor{mycolor}{rgb}{1, 0.0, 0.0}

\newcommand{\sinq}[2]{\stackrel{\mbox{\tiny{#1}}}{#2}}

\newcommand{\R}{\mathbb{R}}
\newcommand{\etal}{{\it{et al.}}}

\usepackage{authblk}

\begin{document}
\title{\textbf{Hardness of Approximation of Euclidean $k$-Median}}

\author[1]{Anup Bhattacharya}
\author[2]{Dishant Goyal}
\author[2]{Ragesh Jaiswal }
\affil[1]{
{\small
Indian Statistical Institute Kolkata,
{
\tt bhattacharya.anup@gmail.com}
}}

\affil[2]{
{\small Indian Institute of Technology Delhi,
{\tt \{dishant.goyal, rjaiswal\}@cse.iitd.ac.in}
}
}


\date{}

\maketitle

\begin{abstract}
    The Euclidean $k$-median problem is defined in the following manner: given a set $\X$ of $n$ points in $\mathbb{R}^{d}$, and an integer $k$, find a set $C \subset \mathbb{R}^{d}$ of $k$ points (called centers) such that the cost function $\Phi(C,\X) \equiv \sum_{x \in \X} \min_{c \in C}  \|x-c\|_{2}$ is minimized.
    The Euclidean $k$-means problem is defined similarly by replacing the distance with squared distance in the cost function. Various hardness of approximation results are known for the Euclidean $k$-means problem \cite{hardness:acks15,hardness:2017_lee_schimdt,hardness:addad19}. 
    However, no hardness of approximation results were known for the Euclidean $k$-median problem. In this work, assuming the unique games conjecture (UGC), we provide the first hardness of approximation result for the Euclidean $k$-median problem. This solves an open question posed explicitly in the work of Awasthi \etal~\cite{hardness:acks15}.
    
    Furthermore, we study the hardness of approximation for the Euclidean $k$-means/$k$-median problems in the bi-criteria setting where an algorithm is allowed to choose more than $k$ centers. That is, bi-criteria approximation algorithms are allowed to output $\beta k$ centers (for constant $\beta>1$) and the approximation ratio is computed with respect to the optimal $k$-means/$k$-median cost.
    In this setting, we show the first hardness of approximation result for the Euclidean $k$-median problem for any $\beta < 1.015$, assuming UGC. We also show a similar bi-criteria hardness of approximation result for the Euclidean $k$-means problem with a stronger bound of $\beta < 1.28$, again assuming UGC. 
\end{abstract}

\newpage
\section{Introduction} \label{section:introduction}

We start by giving the definition of the Euclidean $k$-median problem.
\begin{definition}[$k$-median]
Given a set $\X$ of $n$ points in $\mathbb{R}^{d}$, and a positive integer $k$, find a set of centers $C \subset \mathbb{R}^{d}$ of size $k$ such that the cost function $\Phi(C,\X) \equiv \sum_{x \in \X} \min_{c \in C} \| x - c\|$ is minimized.
\end{definition}

\noindent 
The Euclidean $k$-means problem is defined similarly by replacing the distance with squared distance in the cost function (i.e., replacing $\| x - c\|$ with $\| x - c\|^{2}$).
These problems are also studied in the discrete setting where the centers are restricted to be chosen from a specific set $L \subset \mathbb{R}^{d}$, also given as input.
This is known as the \emph{discrete} version whereas the former version (with $L = \mathbb{R}^{d}$) is known as the \emph{continuous} version. 
\footnote{In the approximation setting, the continuous version is not harder than its discrete counterpart since it is known(e.g., \cite{fixed_k:feldman07,fixed_k:mat2000}) that an $\alpha$-approximation for the discrete problem gives an $\alpha + \veps$ approximation for the continuous version, for arbitrary small constant $\veps>0$. 
}
In this work, we discuss only the continuous version of the problem.
Henceforth, we will refer to the continuous Euclidean $k$-median/$k$-means problem as the Euclidean $k$-median/$k$-means problem or simply the $k$-median/$k$-means problem.

The relevance of the $k$-means and $k$-median problems in various computational domains such as resource allocation, big data analysis, pattern mining, and data compression is well known.
A significant amount of work has been done to understand the computational aspects of the $k$-means/median problems.
The $k$-means problem is known to be $\mathsf{NP}$-hard even for fixed $k$ or $d$~\cite{hardness:aloise09,hardness:das08,hardness:MNV09,hardness:V09}. 
Similar $\mathsf{NP}$ hardness result is also known for the $k$-median problem~\cite{Megiddo}.
Even the $1$-median problem, popularly known as the \emph{Fermat-Weber} problem~\cite{survey:facility_location_2001_drezner}, is a hard problem and designing efficient algorithms for this problem is a separate line of research in itself -- see for e.g.~\cite{fermat:algorithm_1659_torrcelli,fermat:algorithm_1937_Weiszfeld,fermat:algorithm_1989_Chandrasekaran,fermat:algorithm_2002_Badoiu_Har_Peled,fermat:algorithm_2016_Cohen_STOC}.
These hardness barriers motivates approximation algorithms for these problems and a lot of advancement has been made in this area.
For example, there are various polynomial time approximation schemes (PTASs) known for $k$-means and $k$-median when $k$ is fixed (or constant)~\cite{fixed_k:mat2000,fixed_k:kss02,fixed_k:feldman07,fixed_k:chen09,fixed_k:jks14}.
Similarly, various PTASs are known for fixed $d$ ~\cite{fixed_d:2016_addad_FOCS,fixed_d:2016_friggstad_FOCS,fixed_d:2018_vincent_SODA}. 
Various constant factor approximation algorithms are known for $k$-means and $k$-median even considering $k$ and $d$ as part of the input instead of fixed constants.
For the $k$-means problem, constant approximation algorithms have been given~\cite{kmeans:kanungo02,kmeans:Svensson17}, the best being a $6.357$ approximation algorithm by Ahmadian {\it et al.}~\cite{kmeans:Svensson17}.
On the negative side, the $k$-means problem is $\mathsf{NP}$-hard to approximate within any factor smaller than a particular constant greater than one~\cite{hardness:acks15,hardness:2017_lee_schimdt,hardness:addad19}. 
In other words, there exist a constant $\veps>0$ such that there does not exist an efficient $(1+\veps)$-approximation algorithm for the $k$-means problem, assuming $\mathsf{P} \neq \mathsf{NP}$. 
The best-known hardness of approximation result for the problem is $1.07$ due to Addad and Karthik~\cite{hardness:addad19}.
Constant factor approximation algorithms for the $k$-median problem are known~\cite{kmedian:1999_charikar,kmedian:2002_naveen,kmedian:2013_Svensson,kmedian:2015_byrka,kmeans:Svensson17}. 
The best known approximation guarantee for $k$-median is $2.633$ due to Ahmadian~\emph{et al.}~\cite{kmeans:Svensson17}. 
However, unlike the $k$-means problem, no hardness of approximation result was known for $k$-median problem. 
In fact, hardness of approximation for the $k$-median problem was left as an open problem in the work of Awasthi {\it et al.}~\cite{hardness:acks15} who proved the hardness of approximation for the $k$-means problem.
In this work, we solve this open problem by obtaining hardness of approximation result for the Euclidean $k$-median problem assuming that the Unique Games Conjecture holds.
Following is one of the main results of this work.

\begin{theorem}[Main Theorem] \label{theorem:main_theorem}
There exist a constant $\veps > 0$ such that the
Euclidean $k$-median problem cannot be approximated to a factor better than $(1 + \veps)$ assuming the Unique Games Conjecture. 
\end{theorem}

\noindent
\underline{\bf Important note}: We would like to note that similar hardness of approximation result for the Euclidean $k$-median problem using different techniques has been obtained independently by {\it Vincent Cohen-Addad, Karthik C. S., and Euiwoong Lee}.
We came to know about their results through personal communication with the authors. Since their manuscript has not been published online yet, we are not able to add a citation to their work.

Now having hardness of approximation for $k$-means and $k$-median, the next natural step in the {\em beyond worst-case} discussion is to allow more flexibility to the algorithm.
One possible relaxation is to allow an approximation algorithm to choose more than $k$ centers. 
In other words, allow the algorithm to choose $\beta k$ centers (for some constant $\beta > 1$) and produce a solution that is close to the optimal solution with respect to $k$ centers.
This is known as bi-criteria approximation and the following definition formalizes this notion.

\begin{definition}[$(\alpha,\beta)$-approximation algorithm]
An algorithm $\mathcal{A}$ is called an $(\alpha,\beta)$-approximation algorithm for the Euclidean $k$-means/median problem if given any instance $\mathcal{I} = (\X,k)$ with $\X \subset \mathbb{R}^{d}$,  $\mathcal{A}$ outputs a center set $F \subset \mathbb{R}^{d}$ of size $\beta k$ that has the cost at most $\alpha$ times the optimal cost with $k$ centers. That is,
$$
\sum_{x \in \X} \min_{f \in F} \{ D(x, f) \} \leq \alpha \cdot \min\limits_{\substack{
C \subseteq \mathbb{R}^{d} \\
|C| = k
}
} \ \ \left\{ \sum_{x \in \X} \min_{c \in C} \{D(x, c) \} \right\}
$$
For the Euclidean $k$-means problem, $D(p, q) \equiv \|p - q\|^2$ and for the $k$-median problem $D(p, q) \equiv \|p-q\|$.
\end{definition}

One expects that as $\beta$ grows, there exists efficient $(\alpha, \beta)$-approximation algorithms with smaller value of $\alpha$.
This is indeed observed in the work of Makarychev~\emph{et al.}~\cite{bicriteria:2015_Makarychev_Approx_random}. 
For example, their algorithm gives a $(9+\veps)$ approximation for $\beta = 1$; $2.59$ approximation for $\beta = 2$; $1.4$ approximation for $\beta = 3$.
The approximation factor of their algorithm decreases as the value of $\beta$ increases. 
Furthermore, their algorithm gives a $(1+\veps)$-approximation guarantee with $O(k\log(1/\veps))$ centers.
Bandyapadhyay and Varadarajan~\cite{bicriteria:2016_Bandyapadhyay} gave a $(1+\veps)$ approximation algorithm that outputs
$(1+\veps)k$ centers in constant dimension. 
There are various other bi-criteria approximation algorithms that use distance-based sampling techniques and achieve better approximation guarantees than their non-bi-criteria counterparts~\cite{streaming:2009_kmeanspp_ragesh,bicriteria:2009_amit_deshpande,bicriteria:kmeanspp_NIPS_Dennis_Wei}. 
Unfortunately in these bi-criteria algorithms, at least one of $\alpha, \beta$ is large.
Ideally, we would like to obtain a PTAS with a small violation of the number of output centers. 
More specifically, we would like to address the following question:

\begin{quote}
    \emph{Does the $k$-median or $k$-means problem admit an efficient $(1+\veps, 1+\veps)$-approximation algorithm?}
\end{quote}
\noindent Note that such type of bi-criteria approximation algorithms that outputs $(1+\veps)k$ centers have been extremely useful in obtaining a constant approximation for the {\em capacitated} $k$-median problem~\cite{capacitated:kmedian_2017_Li,capacitated:kmedian_2017_Li_uniform} for which no true constant approximation is known yet~\footnote{In the capacitated $k$-median/$k$-means problem there is an additional constraint on each center that it cannot serve more than a specified number of clients (or points).}. Therefore, the above question is worth exploring.
Note that here we are specifically aiming for a PTAS since the $k$-means and $k$-median problems already admit a constant factor approximation algorithm. 
In this work, we give a negative answer to the above question by showing that there exists a constant $\veps> 0$ such that an  efficient $(1+\veps,1+\veps)$-approximation algorithm for the $k$-means and $k$-median problems does not exist assuming the Unique Games Conjecture.
The following two theorems state this result more formally.

\begin{theorem}[$k$-median]\label{theorem:bicriteria_kmedian}
For any constant $1 < \beta < 1.015$, there exists a constant $\veps > 0 $ such that there is no $(1+\veps,\beta)$-approximation algorithm for 
the $k$-median problem assuming the Unique Games Conjecture. 
\end{theorem}

\begin{theorem}[$k$-means]\label{theorem:bicriteria_kmeans}
For any constant $1 < \beta < 1.28$, there exists a constant $\veps > 0 $ such that there is no $(1+\veps,\beta)$-approximation algorithm for the $k$-means problem assuming the Unique Games Conjecture. Moreover, the same result holds for any $1 < \beta < 1.1$ under the assumption that $\mathsf{P} \neq \mathsf{NP}$. 
\end{theorem}

\noindent
\underline{\it Dimensionality reduction}: Note that we can use dimensionality reduction techniques of Makarychev {\it et al.}~\cite{makarychev19} to show that our hardness of approximation results hold for $O(\log{\frac{k}{\veps}}/\veps^2)$ dimensional instances.

In the next subsection, we discuss the known results on hardness of approximation of the $k$-means and $k$-median problems in more detail.


\subsection{Related Work}
\label{section:related_work}
The first hardness of approximation result for the Euclidean $k$-means problem was given by Awasthi~\emph{et al.}~\cite{hardness:acks15}. 
They obtained their result using a reduction from \vc on triangle-free graphs of bounded degree $\Delta$ to the Euclidean $k$-means instances. 
Their reduction yields a $(1+\frac{\veps}{\Delta})$ hardness factor for the $k$-means problem for a particular constant $\veps>0$. 
However, due to an unspecified value of $\Delta$, the authors did not deduce the exact hardness factor for the $k$-means problem. 
To overcome this barrier of unspecified bounded degree, Lee~\emph{et al.}~\cite{hardness:2017_lee_schimdt} showed the hardness of approximation of \vc on triangle-free graphs of bounded degree four. 
Using $\Delta=4$, they obtained a $1.0013$ hardness of approximation for the Euclidean $k$-means problem. 
Subsequently, Addad and Karthik~\cite{hardness:addad19} improved the hardness of approximation to $1.07$ using a reduction from the {\em vertex coverage problem} instead of a reduction from the vertex cover problem. 
Moreover, they also gave several improved hardness results for the discrete $k$-means/$k$-median problems in general and $\ell_{p}$ metric spaces. 
In their more recent work, they also improved the hardness of approximation results for the continuous $k$-means/$k$-median problem in general metric spaces~\cite{hardness:2020_Addad_Continuous_Metric}. 

Unlike the Euclidean $k$-means problem, no hardness of approximation result was known for the Euclidean $k$-median problem.
In this work, we give hardness of approximation result for the Euclidean $k$-median problem assuming the Unique Game Conjecture. 
However, we do not deduce the exact constant hardness factor since we use the same reduction as in ~\cite{hardness:acks15} and hence run into the same problem of unspecified degree that we discussed in the previous paragraph.
As mentioned earlier, in an unpublished work communicated to us through personal communication, {\em Vincent Cohen-Addad, Karthik C. S., and Euiwoong Lee} have independently obtained hardness of approximation result for the Euclidean $k$-median problem using different set of techniques.
They also gave bi-criteria hardness of approximation results in $\ell_{\infty}$-metric for the $2$-means and $2$-median problems.
We would like to point out that in the bi-criteria setting, our result is the first hardness of approximation result for the Euclidean $k$-means/$k$-median problem to the best of our knowledge.

All of our hardness of approximation results are based on the reduction from \vc on bounded degree and triangle-free graphs. 
As we mentioned earlier, the same reduction is used in~\cite{hardness:acks15,hardness:2017_lee_schimdt} to obtain the hardness of approximation for the Euclidean $k$-means problem. 
However, extending this gap-preserving reduction to the Euclidean $k$-median setting is non-trivial.
This problem was left as an open problem by Awasthi {\it et al.}~\cite{hardness:acks15}.
In the next subsection, we discuss this reduction and related difficulties in obtaining the hardness of approximation for Euclidean $k$-median problem.

\subsection{Comparison with~\cite{hardness:acks15} and Technical Contribution}\label{subsection:kmeans_comparison}
Awasthi \emph{et al.}~\cite{hardness:acks15} showed that the $k$-means problem is hard to approximate within a factor $(1+\veps)$ for a particular constant $0 < \veps < 1)$. 
We borrow some techniques from their work and show the hardness of approximation for the $k$-median problem, within a factor $(1+\veps')$ for a particular constant $0 < \veps' < 1$. 
However, this task is challenging and non-trivial. 
We will discuss these challenges in this subsection. 
First, let us briefly discuss the results and techniques of \cite{hardness:acks15}.

Awasthi \etal~\cite{hardness:acks15} first gave a $(1+\veps)$-approximation preserving reduction from \vc on bounded degree graphs to \vc on bounded degree \emph{triangle-free} graphs. 
Then, they gave a reduction from \vc on bounded degree triangle-free graphs to the Euclidean $k$-means instances. 
The first reduction straightaway gives a $1.36$ hardness of approximation for the \vc on bounded degree and triangle-free graphs since a $1.36$ hardness of approximation is already known for the \vc on bounded degree graphs~\cite{Vertex_Cover:2005_Dinur_Safra}.
Following is a formal statement for this.
\begin{theorem}[Corollary 5.3~\cite{hardness:acks15}]\label{theorem:VC_PNP}
Given any unweighted bounded degree, triangle-free graph $G$, it is
$\mathsf{NP}$-hard to approximate \vc within any factor smaller than $1.36$.
\end{theorem}

\noindent Here is the description of the second reduction.
\begin{quote}
\noindent \textit{Construction of $k$-means instance:}
\noindent Let $(G, k)$ be a hard \vc instance where $G$ has bounded degree $\Delta$. 
Let $n$ denote the number of vertices in the graph and $m$ denote the number of edges in the graph. 
A $k$-means instance $\mathcal{I} \coloneqq (\X,k)$ with $\X \subset \R^n$ is constructed as follows. 
For every vertex $i \in V$, there is an $n$-dimensional vector $x_{i} \coloneqq (0,\dotsc,1,\dotsc,0)$ in $\{0,1\}^{n}$, which has $1$ at $i^{th}$ coordinate and $0$ at the rest of the coordinates. For each edge $e = (i,j) \in E$, a point $x_{e} \coloneqq x_{i} + x_{j}$ is defined in $\{0,1\}^{n}$. The point set $\X \coloneqq \{ x_{e} \mid e \in E\}$ and parameter $k$ defines the $k$-means instance.
\end{quote}
The following theorem based on the above construction is given in \cite{hardness:acks15}.

\begin{theorem} [Theorem 4.1~\cite{hardness:acks15}]\label{theorem:reduction_kmeans}
There is an efficient reduction from \vc on bounded degree, triangle-free
graphs to the Euclidean $k$-means instances that satisfies the following properties:
\begin{enumerate}
    \item If the \vc instance has value $k$, then the $k$-means instance has a cost at most $(m-k)$.
    \item If the \vc instance has value at least $(1 + \veps) \cdot k$, then the optimal cost of $k$-means instance is at least $(m - k + \delta k)$. 
\end{enumerate}
Here, $\veps$ is some fixed constant $> 0$ and $\delta = \Omega(\veps)$
\end{theorem}

\noindent The above reduction is only valid for the bounded degree triangle-free graphs.
This is the main reason, why the \vc was shown to be $\mathsf{APX}$-hard on these graph instances. 
Furthermore, the above theorem implies that the $k$-means problem is $\mathsf{APX}$-hard. 
Following is a formal statement for the same (see Section 4 of~\cite{hardness:acks15} for the proof of this result).

\begin{corollary}
There exist a constant $\veps' > 0$ such that it is $\mathsf{NP}$-hard to approximate the Euclidean $k$-means problem to any factor better than $(1 + \veps')$.
\end{corollary}

\noindent Note that we can further reduce the above $k$-means instances to $k$-means instances in a smaller dimensional space by applying standard dimensionality reduction techniques~\cite{Johnson_Lindenstrauss,makarychev19}. 
Finally, the authors conclude with the following important question (see Section 6 of~\cite{hardness:acks15}):
\begin{quote}
    ``\emph{It would also be interesting
to study whether our techniques give hardness of approximation results for the Euclidean k-median problem.}''
\end{quote}
In other words, if we employ the same construction as in Awasthi \etal~\cite{hardness:acks15}, and let $\mathcal{I} = (\X,k)$ denote the Euclidean $k$-median instance, then can we show this reduction to be gap preserving for $k$-median? This question is challenging due to the hardness of the $1$-median problem.
Unlike the $1$-means problem, where the optimal center is the centroid of the point set, the $1$-median problem does not have any closed-form expression for the optimal center. As we mentioned earlier, this problem is also popularly known as the \emph{Fermat Weber} problem~\cite{survey:facility_location_2001_drezner}. However, despite these difficulties, we can show that the above reduction is gap-preserving for Euclidean $k$-median problem. This is made possible using the idea that we do not require the exact optimal cost of $1$-median instance, instead, good lower and upper bounds on the optimal $1$-median cost suffice for this problem. Following are the two main ideas that we use here. In order to obtain an upper bound, we simply compute the $1$-median cost of a point set with respect to its centroid. In order to obtain a lower bound, we use a clever decomposition technique that decomposes a $1$-median instance into many smaller size instances and bound the total cost in terms of the cost of simpler instances, the $1$-median costs of which can be easily computed. We elaborate on these ideas later in Section~\ref{section:vertex_cover}. Overall, the novelty of this work lies in bounding the optimal cost of the $k$-median instances and further deducing its relationship with the vertex cover of the bounded degree triangle-free graphs. 
Furthermore, we extend these techniques to show the bi-criteria hardness of approximation results of Euclidean $k$-means and $k$-median problems.

\section{Useful Facts and Inequalities} \label{subsection:basic_facts}
In this section, we discuss some basic facts and inequalities that we will frequently use in our proofs. 
First, we note that the Fermat-Weber problem is not difficult for all $1$-median instances. 
We can efficiently obtain $1$-median for some special instances. 
For example, for a set of equidistant points, the $1$-median is simply the centroid of the point set. 
We give a proof of this statement in the next section.
Most importantly, we use the following two facts to compute the $1$-median cost.

\begin{fact}[\cite{fermat:1987_ducharme}]\label{fact:unique_median}
For a set of non-collinear points the optimal $1$-median is unique.
\end{fact}

\begin{fact}[\cite{fermat:1991_Lopuha,fermat:1942_haldane}]\label{fact:isometric}
The $1$-median cost is preserved if pairwise distances between the input points are preserved.\footnote{Even though this statement is not explicitly mentioned in these references, it can be derived from them.}
\end{fact}


We use the above fact, in vector spaces where it is tricky to compute the optimal $1$-median exactly.
In such cases, we transform the space to a different vector space, where computing the $1$-median is relatively simpler. 
More specifically, we employ a rigid transformation since it preserves pairwise distances.
Next, we give a simple lemma, that is used to prove various bounds related to the quantity $\sqrt{m(m-1)}$.
\begin{lemma}\label{lemma:basic_bound}
Let $m$ and $t$ be any positive real numbers greater than one. If $m \geq t$, the following bound holds:
\[
m - (t-\sqrt{t(t-1)}) \leq \sqrt{m(m-1)} \leq m-1/2
\]
\end{lemma}
\begin{proof}
The upper bound follows from the following sequence of inequalities:
\[
\sqrt{m(m-1)} < \sqrt{m^{2}-m+1/4} = \sqrt{(m-1/2)^{2}} = m-1/2
\]  
The lower bound follows from the following sequence of inequalities:
    \begin{align*}
    \sqrt{m(m-1)} &= m + (\sqrt{m(m-1)}-m) \\
    &= m + m \cdot \left( \sqrt{\frac{m-1}{m}} -1 \right) \\
    &\geq m + t \cdot \left( \sqrt{\frac{t-1}{t}} -1 \right) \quad \quad \textrm{$\because \frac{a+1}{b+1} \geq \frac{a}{b}$ for $b \geq a$ }\\
    &= m - (t-\sqrt{t(t-1)})
    \end{align*}
\noindent This completes the proof of the lemma.
\end{proof}
\noindent In the next section, we show the inapproximability of the Euclidean $k$-median problem, assuming the Unique Games Conjecture (UGC). 


\section{Inapproximability of Euclidean $k$-Median}\label{section:kmedian_inapproximability}

We show a gap preserving reduction from \vc on bounded degree triangle-free graphs to the Euclidean $k$-median instances. 
In addition to it, we use the following result which follows from~\cite{hardness:acks15} and~\cite{Vertex_Cover:2010_Bounded_Degree_Graphs_Khot}~\footnote{\cite{Vertex_Cover:2010_Bounded_Degree_Graphs_Khot} showed that \vc on $d$-degree graphs is hard to approximate within any factor smaller than $2-\veps$, for $\veps = (2+o_{d}(1)) \cdot \frac{\log \log d}{\log d}$ assuming the Unique Games Conjecture. Therefore, $\veps$ can be set to arbitrarily small value by taking sufficiently large value of $d$.}.
\begin{theorem}\label{theorem:VC_UGC}
Given any unweighted triangle-free graph $G$ of bounded degree, Vertex Cover can not be approximated within a factor smaller than $2 - \veps$, for any constant $\veps > 0$, assuming the Unique Games Conjecture.
\end{theorem}

In Section~\ref{subsection:kmeans_comparison}, we described the construction of a $k$-means instance from a \vc instance. We use the same construction for the $k$-median instances. Let $G = (V,E)$ denote a triangle-free graph of bounded degree $\Delta$. Let $\mathcal{I} = (\X,k)$ denote the Euclidean $k$-median instance constructed from $G$.
We establish the following theorem based on this construction.
\begin{theorem}\label{theorem:reduction_kmedian}
There is an efficient reduction from instances of \vc on triangle-free
graphs with $m$ edges to those of Euclidean $k$-median that satisfies the following properties:
\begin{enumerate}
    \item If the graph has a vertex cover of size $k$, then the $k$-median instance has a solution of cost at most $m-k/2$
    \item If the graph has no vertex cover of size $\leq (2 - \veps)\cdot k$, then the cost of any $k$-median solution on the instance is at least $m-k/2+\delta k$ 
\end{enumerate}
Here, $\veps$ is some fixed constant and $\delta = \Omega(\veps)$.
\end{theorem}

The graphs with a vertex cover size at most $k$ are said to be ``Yes'' instances and the graphs with no vertex cover of size $\leq (2-\veps)k$ are said to be ``No'' instances. Now, the above theorem gives the following inapproximability result for the Euclidean $k$-median problem.

\begin{corollary} \label{corollary:reduction_kmedian}
There exists a constant $\veps' > 0$ such that the Euclidean $k$-median problem can not be approximated to a factor better than $(1 + \veps')$, assuming the Unique Games Conjecture. 
\end{corollary}
\begin{proof}

Since the hard \vc instances have bounded degree $\Delta$, the maximum matching of such graphs is at least $\lceil \frac{m}{2\Delta}\rceil$. First, let us prove this statement. Suppose $M$ be a matching, that is initially empty, i.e., $M = \emptyset$.
We construct $M$ in an iterative manner. First, we pick an arbitrary edge from the graph and add it to $M$. Then, we remove this edge and all the edges incident on it. We repeat this process for the remaining graph till the graph contains no edge.
In each iteration, we remove at most $2 \Delta$ edges. 
Therefore, the matching size of the graph is at least $\lceil \frac{m}{2\Delta}\rceil$. 

Now, suppose $k < \frac{m}{2\Delta}$. Then, the graph does not have a vertex cover of size $k$ since matching size is at least $\lceil \frac{m}{2\Delta}\rceil$.
Therefore, such graph instances can be classified as ``No'' instances in polynomial time. 
So, they are not the hard \vc instances.
Therefore, we can assume $k \geq \frac{m}{2\Delta}$ for all the hard \vc instances. 
In that case, the second property of Theorem~\ref{theorem:reduction_kmedian}, implies that the cost of $k$-median instance is $(m - \frac{k}{2}) + \delta k \geq (1+ \frac{\delta}{2\Delta})\cdot (m - \frac{k}{2})$. 
Thus, the $k$-median problem can not be approximated within any factor smaller than $1 + \frac{\delta}{2\Delta} = 1 + \Omega(\veps)$.
\end{proof}

Let us define some notations before proving Theorem~\ref{theorem:reduction_kmedian}. 
For any subgraph $S$ of the graph $G$, we denote its number of edges by $m(S)$. 
Recall that a point in $\X$ corresponds to an edge of the graph. Therefore, a subgraph $S$ of $G$  defines a subset of points $\X(S) \coloneqq \{ x_{e} \mid e \in E(S) \}$ of $\X$. 
We define the $1$-median cost of $\X(S)$ with respect to a center $c \in \mathbb{R}^{n}$ as:
\[
\Phi(c,S) \equiv \sum_{x \in \X} d(x,c). 
\]

Furthermore, we define the  \textbf{optimal} $1$-median cost of $\X(S)$ as $\Phi^{*}(S)$, i.e.,
\[
\Phi^{*}(S) \equiv \min_{c \in \mathbb{R}^{n}} \Phi(c,S)
\]

In the discussion that follows, we often use the statement: ``optimal $1$-median cost of a graph $S$'', which simply means: ``optimal $1$-median cost of the cluster $\X(S)$''.

\subsection{Completeness} \label{section:completeness}

Let $V = \{v_{1},\dotsc,v_{k}\}$ be a vertex cover of $G$.  Let $S_{i}$ denote the set of edges covered by $v_{i}$. If an edge is covered by two vertices $i$ and $j$, then we arbitrarily keep the edge either in $S_{i}$ or $S_{j}$. Let $m_{i}$ denote the number of edges in $S_{i}$. We define $\{\X(S_{1}),\dotsc,\X(S_{k})\}$ as a clustering of the point set $\X$. Now, we show that the cost of this clustering is at most $m-k/2$. 
Note that each $S_{i}$ forms a star graph centered at $v_{i}$. Moreover, the point set $\X(S_{i})$ forms a regular simplex of side length $\sqrt{2}$.
We compute the optimal cost of $\X(S_{i})$ using the following lemma.


\begin{lemma}\label{lemma:simplex_cost}
For a regular simplex on $r$ vertices and side length $s$, the optimal $1$-median is the centroid of the simplex. 
Moreover, the optimal $1$-median cost is $s \cdot \sqrt{\fdfrac{r(r-1)}{2}}$.
\end{lemma}
\begin{proof}
The statement is simple for $r=1, 2$. So for the rest of the proof, we assume that $r > 2$.
Suppose $A = \{ a_{1}, a_{2},\dotsc, a_{r}\}$ denote the vertex set of a regular simplex. Let $s$ be the side length of the simplex. Using Fact~\ref{fact:isometric}, we can represent each point $a_{i}$ in an $r$-dimensional space as follows:
\[
a_{1} \coloneqq \left( \frac{s}{\sqrt{2}},0,...,0 \right),\ \ a_{2} \coloneqq \left( 0,\frac{s}{\sqrt{2}}, ...,0 \right),\ \ \dotsc, \ \ 
a_{r} \coloneqq \left( 0, 0, ..., \frac{s}{\sqrt{2}}\,\right)
\]
Note that the distance between any $a_{i}$ and $a_{j}$ is $s$, which is the side length of the simplex. 
Let $c^{*} = (c_{1},\dotsc,c_{r})$ be an optimal $1$-median of $A$. Then, the $1$-median cost is the following:
\[
\Phi(c^{*},A) = \sum_{i = 1}^{r} \| a_{i} - c^{*}\| = \sum_{i = 1}^{r} \left( \sum_{j = 1}^{r} c_{j}^{2} - c_{i}^{2} + \left(\frac{s}{\sqrt{2}} - c_{i}\right)^{2} \right)^{1/2}
\]
Suppose $c_{i} \neq c_{j}$ for any $i \neq j$. Then, we can swap $c_{i}$ and $c_{j}$ to create a different median, while keeping the $1$-median cost the same. It contradicts the fact that there is only one optimal $1$-median, by Fact~\ref{fact:unique_median}.
Therefore, we can assume $c^{*} = (c,c,\dotsc,c)$. 
Now, the optimal $1$-median cost is:
\[
\Phi^*(A) = \Phi(c^{*},A) \coloneqq r \cdot \sqrt{ \left( c-\frac{s}{\sqrt{2}} \right)^2 + (r-1)\cdot c^2 }
\]
The function $\Phi(c^{*},A)$ is strictly convex and attains minimum at $c = \fdfrac{s}{m \cdot \sqrt{2}}$, which is the centroid of $A$. The optimal $1$-median clustering cost is $\Phi(c^{*},A) = s \cdot \sqrt{\fdfrac{r(r-1)}{2}}$. This completes the proof of the lemma.
\end{proof}

\noindent The following corollary establishes the cost of a star graph $S_{i}$. 
\begin{corollary}\label{corollary:cost_star}
Any star graph $S_{i}$ with $r$ edges has the optimal $1$-median cost of $\sqrt{r(r-1)}$
\end{corollary}
Furthermore, note that a set of $r$ pairwise vertex-disjoint edges forms a regular simplex in $\X$, of side length $2$. 
The following corollary establishes the cost of such clusters. 
\begin{corollary}\label{corollary:cost_non_star}
Let $F$ be any non-star graph with $r$ pairwise vertex-disjoint edges, then the optimal $1$-median cost of $F$ is $\sqrt{2} \cdot \sqrt{r(r-1)}$
\end{corollary}

We use Corollary~\ref{corollary:cost_non_star} in Section~\ref{section:vertex_cover}. For now, we only use Corollary~\ref{corollary:cost_star} to bound the optimal $k$-median cost of $\X$. Let $OPT(\X,k)$ denote the optimal $k$-median cost of $\X$. The following sequence of inequalities proves the first property of Theorem~\ref{theorem:reduction_kmedian}.
\[
OPT(\X,k) \leq \sum_{i= 1}^{k} \Phi^{*}(S_{i}) \stackrel{\tiny{(Corollary~\ref{corollary:cost_star})}}{=} \sum_{i = 1}^{k} \sqrt{m_{i}(m_{i}-1)} \stackrel{\tiny{(Lemma~\ref{lemma:basic_bound})}}{\leq} \sum_{i = 1}^{k} \left( m_{i} - \frac{1}{2} \right) = m - \frac{k}{2}.
\]

\subsection{Soundness}\label{section:soundness}
Now, we prove the second property of Theorem~\ref{theorem:reduction_kmedian}. 
For this, we prove the equivalent contrapositive statement:
If the optimal $k$-median clustering of $\X$ has cost at most $\left(m - \frac{k}{2} + \delta k \right)$, for some constant $\delta >0$, then $G$ has a vertex cover of size at most $(2-\veps)k$, for some constant $\veps>0$.  
Let $\mathcal{C}$ denote an optimal $k$-median clustering of $\X$. 
We classify its optimal clusters into two categories: (1) \textit{star} and (2) \textit{non-star}. 
Let $ F_{1},F_{2},\dotsc,F_{t}$ denote the non-star clusters, and $S_{1},\dotsc,S_{k-t}$ denote the star clusters. For any star cluster, the vertex cover size is exactly one. 
Moreover, the optimal $1$-median cost of any star cluster with $r$ edges is $\sqrt{r(r-1)}$. 
On the other hand, it may be tricky to compute the vertex cover or the optimal cost of any non-star cluster exactly. Suppose the optimal $1$-median cost of a non-star cluster $F$ on $r$ edges is given as $\sqrt{r(r-1)} + \delta(F)$. We denote by $\delta(F)$ the \emph{extra-cost} of a cluster $F$. 
In the entire discussion, we will use $|.|$ to denote the number of edges in a given graph.
When used in the context of a set, $|.|$ denotes the cardinality of the given set.
Using this, we write:
\[
\delta(F) \equiv \Phi^*(F) - \sqrt{|F|(|F|-1)}
\]
The following lemmas bound the vertex cover of $F$ in terms of $\delta(F)$.

\begin{lemma}\label{lemma:non_star_vertex_cover_1}
Any non-star cluster $F$ with a maximum matching of size two has a vertex cover of size at most $1.62 + \del$.
\end{lemma}

\begin{lemma}\label{lemma:non_star_vertex_cover_2}
Any non-star cluster $F$ with a maximum matching of size at least three has a vertex cover of size at most $1.8 + \del$. 
\end{lemma}

These lemmas are the key to proving the main result. 
We will prove these lemmas later. 
First, let us see how they give a vertex cover of size at most $(2-\veps)k$. 
Let us classify the star clusters into the following two sub-categories:
\begin{enumerate}
    \item[(a)] Clusters composed of exactly one edge. Let these clusters be: $P_{1},P_{2},\dotsc,P_{t_{1}}$. 
    \item[(b)] Clusters composed of at least two edges. Let these clusters be: $S_{1},S_{2},\dotsc,S_{t_{2}}$.
\end{enumerate}
    
\noindent Similarly, we classify the non-star clusters into the following two sub-categories:
\begin{enumerate}
    \item[(i)] Clusters with a maximum matching of size two. Let these clusters be: $W_{1}, W_{2}, \dotsc,W_{t_{3}}$ 
    \item[(ii)] Clusters with a maximum matching of size at least three. Let these clusters be: $Y_{1}, Y_{2}, \dotsc,Y_{t_{4}}$  
\end{enumerate}

\noindent Note that $t_{1} + t_{2} + t_{3} + t_{4}$ equals $k$. Now, consider the following strategy of computing the vertex cover of $G$. Suppose, we compute the vertex cover for every cluster separately. Let $C_{i}$ be any cluster, and $|VC(C_{i})|$ denote the vertex cover size of $C_{i}$. Then, the vertex cover of $G$ can be simply bounded in the following manner:
\[
|VC(G)| \leq  \sum_{i = 1}^{t_{1}}|VC(P_{i})| + \sum_{i = 1}^{t_{2}}|VC(S_{i})| + \sum_{i = 1}^{t_{3}}|VC(W_{i})| + \sum_{i = 1}^{t_{4}}|VC(Y_{i})|
\]
However, we can obtain a vertex cover of smaller size using a slightly different strategy. 
In this strategy, we first compute a minimum vertex cover of all the clusters except single edge clusters $P_{1},P_{2},\dotsc,P_{t_{1}}$.
Suppose that vertex cover is $VC'$. 
Then we compute a vertex cover for $P_{1},P_{2},\dotsc,P_{t_{1}}$. Now, let us see why this strategy gives a vertex cover of smaller size than before.
Note that some vertices in $VC'$ may also cover the edges in $P_{1},\dotsc,P_{t_{1}}$. 
Suppose there are $t_{1}'$ clusters in $P_{1},\dotsc,P_{t_{1}}$ that remain uncovered by $VC'$. 
Without loss of generality, assume these clusters to be $P_{1},\dotsc,P_{t_{1}'}$. 
Now, the vertex cover of $G$ is bounded in the following manner:
\begin{eqnarray*}
|VC(G)| &\leq&  |VC\left( \cup_{i=1}^{t_1'}P_{i}\right)| + |VC'| \\
&=& |VC\left( \cup_{i=1}^{t_1'}P_{i}\right)| + |VC\left( (\cup_{j=1}^{t_2} S_j) \cup (\cup_{k=1}^{t_3} W_k) \cup (\cup_{l=1}^{t_4} Y_l)\right)|\\
&\leq& |VC\left( \cup_{i=1}^{t_1'}P_{i}\right)| + \sum_{i = 1}^{t_{2}}|VC(S_{i})| + \sum_{i = 1}^{t_{3}}|VC(W_{i})| + \sum_{i = 1}^{t_{4}}|VC(Y_{i})|
\end{eqnarray*}
Now, we will try to bound the size of the vertex cover of $P_{1} \cup ... \cup P_{t_{1}'}$. 
Note that we can cover all these single-edge clusters with $t_{1}'$ vertices by choosing one vertex per cluster.
However, it may be possible to obtain a vertex cover of smaller size if we collectively consider all these clusters. 
Suppose $E_{P}$ denote the set of all edges in $P_{1},\dotsc,P_{t_{1}'}$ and $V_{P}$ denote the vertex set spanned by them. 
We define a graph $G_{P} = (V_{P},E_{P})$. 
Note that $G_{P}$ is a subgraph of $G$.
Further, we define another subgraph $\overline{G}_{P}$ of $G$ that is obtained after removing the edges of $E_{P}$ from $G$. 
That is, $\overline{G}_{P} = (V,E \setminus E_{P})$. 
In other words, $\overline{G}_{P}$ is the graph spanned by the remaining clusters: $S_{1},\dotsc,S_{t_2}$; $W_{1},\dotsc,W_{t_3}$; $Y_{1}, \dotsc,Y_{t_4}$; and $P_{t_1'+1},\dotsc, P_{t_1}$.
An important property of $G_{P}$ is that any edge of $\overline{G}_{P}$ does not have its both endpoints in $V_{P}$. In other words, every edge of $\overline{G}_{P}$ is incident on at most one edge of $G_{P}$. 
This is because every edge of $\overline{G}_{P}$ has at least one endpoint in $VC'$, and $G_{P}$ is only defined on the edge that are not incident on $VC'$. This property will help us in obtaining a better vertex cover for $P_{1},\dotsc,P_{t_{1}'}$. We bound the size of the vertex cover in the following lemma.

\begin{lemma} \label{lemma:single_star_vertex_cover}
Let $\delta > 0$ be any constant and $G_P$ be as defined above.
If $G_P$ does not have a vertex cover of size $\leq (\frac{2t_{1}'}{3} + 8 \delta k)$, then $G$ has a vertex cover of size at most $(2k -2\delta k)$.
\end{lemma}
\begin{proof}
Let $M_{P}$ be a maximal matching of $G_{P}$.
We divide the analysis into two cases based on the size of $M_{P}$. In the first case, we consider $|M_{P}| \leq (t_{1}'/3 + 4\delta k)$ and show that $G_{P}$ has a vertex cover of size at most $(2t_{1}'/3 + 8 \delta k)$. 
In the second case, we consider $|M_{P}| > (t_{1}'/3 + 4\delta k)$ and show that $G$ has a vertex cover of size at most $(2k - 2\delta k)$. 
Let us consider these cases one by one.

\begin{itemize}
    \item \uline{Case-I}:  $\big( \, |M_{P}| \leq t_{1}'/3 + 4\delta k \, \big)$
    
    \noindent In this case, we simply cover all edges of $G_{P}$ by picking both endpoints of every edge in $M_{P}$. 
    Thus, $G_{P}$ can be covered using only $(2t_{1}'/3 + 8\delta k)$ vertices.
    
    \item  \uline{Case-II}: $\big( \, |M_{P}| > t_{1}'/3 + 4 \delta k \, \big)$
    
    In this case, we will incrementally construct a vertex cover $VC_{G}$ of the graph $G$ of size at most $(2k - 2\delta k)$.
    First, let us discuss the main idea of this incremental construction.
    During the construction, we maintain a maximal matching $M_{G}$ of the graph $G$. 
    Then, for every edge in $M_{G}$, we add both its endpoints to $VC_{G}$, except at least $2\delta k$ edges for which we choose only one endpoint in $VC_{G}$. 
    Suppose, $VC_{G}$ covers all edges in $G$. 
    Then, we can claim that $VC_{G}$ has a size at most $(2k - 2\delta k)$.
    Note that for the hard \vc instances, we can assume that maximum matching size is at most $k$. 
    This is because the graphs with a matching of size $> k$ have a minimum vertex cover of size $> k$. 
    Therefore, such instances can be simply classified as ``No'' instances in polynomial time. 
    Since the size of a maximal matching is always less than the size of a maximum matching, we can further assume $|M_{G}| \leq k$ for the hard \vc instances. 
    This implies that $VC_{G}$ has a size at most $2(k-2\delta k)+2\delta k = (2k - 2\delta k)$.
    Now, let us discuss the construction of such a matching $M_G$ and correspondingly the vertex cover $VC_{G}$.
    
    Initially, both $M_{G}$ and $VC_{G}$ are empty sets. That is, $M_{G} = \emptyset$ and $VC_{G} = \emptyset$. Recall that $G_{P}$ is the graph spanned by $E_{P}$ and $\overline{G}_{P}$ is the graph spanned by $E \setminus E_{P}$. Let $E_{I}$ denote the set of edges in $\overline{G}_{P}$ that are incident on $M_{P}$. Based on this, we define two new graph: (1) $G_{R} \coloneqq (V,E_{R})$ where $E_{R} = E_{P} \cup E_{I}$, and (2) $G' \coloneqq (V,E')$ where $E' = E \setminus (E_{P} \cup E_{I})$. In other words, $G_{R}$ is the graph spanned by the edges of $G_{P}$ and the edges of $\overline{G}_{P}$ that are incident on $M_{P}$; and $G'$ is the graph spanned by the edges of $\overline{G}_{P}$ that are not incident on $M_{P}$. 
    Now, we compute a maximal matching $M'$ of $G'$, and execute the following procedure:\vspace{-0.5cm}

\begin{minipage}[t]{\linewidth}
\begin{figure}[H]
\begin{mdframed}
{\tt Procedure 1} \vspace*{1mm} \\
\hspace{1cm} (1) \hspace*{0.1cm} $M_{G} \gets M'$ \\
\hspace{1cm} (2) \hspace*{0.1cm} For every edge $e \equiv (u,v) \in M_{G}$: \\
\hspace{1cm} (3) \hspace*{0.7cm} $VC_{G} \gets VC_{G} \cup \{u,v\}$ \\
\hspace{1cm} (4) \hspace*{0.7cm} Update $G_{R}$ by removing all the edges in them that are incident on $u$ and $v$\\
\vspace*{-3mm}
\end{mdframed}
\vspace*{-4mm}
\caption{}
\label{fig:Procedure_1}
\end{figure}
\end{minipage}
    Note that the above procedure removes every edge of $G'$ since $M'$ is a maximal matching of $G'$. 
    Now, we will find a vertex cover and maximal matching of updated $G_{R}$. Note that any maximal matching of $G_{R}$ can be combined with that of $M_{G}$ to form a maximal matching of the original graph $G$ since we already removed the edges which were incident on $M_{G}$.
    
    Note that $G_{R}$ is composed of the edge sets $E_{P}$ and a subset of edges from $E_{I}$.
    Therefore, $G_{P}$ is also a subgraph of $G_{R}$. Recall that $M_{P}$ is a maximal matching of $G_{P}$.
    We define a new edge set $U_{P}$ that denote the set of unmatched edges in $G_{P}$, i.e., $U_{P} = E_{P} \setminus M_{P}$. 
    Note that $|U_{P}| \leq 2t_{1}'/3 - 4 \delta k$ since $|M_{P}| > t_{1}'/3 + 4 \delta k$. Moreover, every edge in $U_{P}$ is incident on $M_{P}$, since $M_{P}$ is a maximal matching of $G_{P}$. 
    In the next two procedures, we remove some edges from $U_{P}$ and $M_{P}$ such that the updated $U_{P}$ only contains those edges that are incident on one  edge of the updated $M_{P}$. 
    Then, we will use this graph to obtain a vertex cover of $G$ of size at most $2k - 2\delta k$.
    Following is the first procedure:\vspace{-0.5cm}

\begin{minipage}[t]{\linewidth}
\begin{figure}[H]
\begin{mdframed}
{\tt Procedure 2} \vspace*{1mm} \\
\hspace{1cm} (1) \hspace*{0.3cm} 
\textbf{while} there is an edge $e \equiv (u,v) \in M_{P}$ that is incident on two edges of $U_{P}$ \\
\hspace{1cm} (2) \hspace*{0.7cm} $M_{G} \gets M_{G} \cup \{e\}$ \\
\hspace{1cm} (3) \hspace*{0.7cm} $VC_{G} \gets VC_{G} \cup \{u,v\}$ \\
\hspace{1cm} (4) \hspace*{0.7cm} Update $G_{R}$, $M_{P}$, and $U_{P}$ by removing all the edges in them that are \\
\hspace*{1.7cm} incident on $u$ and $v$\\
\vspace*{-3mm}
\end{mdframed}
\vspace*{-4mm}
\caption{}
\label{fig:Procedure_2}
\end{figure}
\end{minipage} 
 
    Note that before the beginning of this procedure there were at most $\leq 2t_{1}'/3 - 4 \delta k$ edges in $U_{P}$ and at least $t_{1}'/3 + 4 \delta k$ edges in $M_{P}$. Then, the above procedure removes at least two edges from $U_{P}$ and  one edge from $M_{P}$ in each iteration of the while loop. Therefore, at the end of the procedure at least $6 \delta k$ edges remain in $M_{P}$.  Moreover, at the end of the procedure, $M_{P}$ has the property that no two edges in $U_{P}$ are incident on the same edge of $M_{P}$. Next, consider the following procedure:\vspace{-0.5cm}
    
    \begin{minipage}[t]{\linewidth}
\begin{figure}[H]
\begin{mdframed}
{\tt Procedure 3} \vspace*{1mm} \\
\hspace{1cm} (1) \hspace*{0.3cm} \textbf{while} there is an edge $e \equiv (u,v) \in U_{P}$ that is incident on two edges $e_{1}, e_{2}  \in M_{P}$ \\
\hspace{1cm} (2) \hspace*{0.7cm} Arbitrarily pick one edge from $\{e_{1},e_{2} \} $. W.l.o.g., let $e_{1} \equiv (u,v)$ be that edge. \\
\hspace{1cm} (3) \hspace*{0.7cm} $M_{G} \gets M_{G} \cup \{e_{1}\}$ \\
\hspace{1cm} (4) \hspace*{0.7cm} $VC_{G} \gets VC_{G} \cup \{u,v\}$ \\
\hspace{1cm} (5) \hspace*{0.7cm} Update $G_{R}$, $M_{P}$, and $U_{P}$ by remove all the edges in them that are \\
\hspace*{1.7cm} incident on $u$ and $v$\\
\vspace*{-3mm}
\end{mdframed}
\vspace*{-4mm}
\caption{}
\label{fig:Procedure_2}
\end{figure}
\end{minipage} 

    Let $p \geq 6\delta k$ denote the number of edges in $M_{P}$, before the beginning of the above procedure.
    Now, observe that the above procedure removes one edge from $U_{P}$ and one edge from $M_{P}$ in each iteration of the while loop. Furthermore, the while loop executes at most $p/2$ times since $M_{P}$ had the property that two edges of $U_{P}$ do not incident on the same edge of $M_{P}$. Therefore, at the end of the procedure, $M_{P}$ contains at least $p/2 \geq 3 \delta k$ edges. 
    Moreover, $U_{P}$ has obtained the property that all its edges are incident on exactly one edge of $M_{P}$. Now, recall that all edges in $E_{I}$ are also incident on exactly one edge of $M_{P}$. We proved this property earlier (just before we stated Lemma~\ref{lemma:single_star_vertex_cover}) for every edge of $\overline{G}_{p}$ that was incident on $G_{p}$. Therefore, at this point, $G_{R}$ consists of the edge set $M_{P}$ of size at least $3 \delta k$, and the edges that are incident on exactly one edge of $M_{P}$. Now, we will find a maximal matching and vertex cover of the remaining graph $G_{R}$.
    
    We color the edges of $M_{P}$ with red color and the remaining edges of $G_{R}$ with blue color. 
    Now, let us define the concept of ``\emph{plank edge}''. 
    A plank edge is red edge $e \equiv (u,v) \in M_{P}$ that satisfies the following two conditions. 
    The first condition is that at least one blue edge in $G_{R}$ is incident on $u$ and at least one incident on $v$. 
    Let $e_{u}$ and $e_{v}$ denote the blue edges incident on $u$ and $v$, respectively. 
    The second condition is that $e_{u}$ and $e_{v}$ should be vertex disjoint from every edge of $M_{G}$ (with respect to the current set $M_G$ which keeps getting updated). 
    In other words, we should be able to add $e_{u}$ and $e_{v}$ to $M_{G}$. 
    Note that $e_{u}$ and $e_{v}$ do not share any common vertex; otherwise it would form a triangle. 
    Therefore, we can add both of them to $M_{G}$.
    Now, we complete the construction of the maximal matching $M_{G}$ using the following procedure.\vspace{-0.5cm}
    
    \begin{minipage}[t]{\linewidth}
\begin{figure}[H]
\begin{mdframed}
{\tt Procedure 4} \vspace*{1mm} \\
\hspace{1cm} (1) \hspace*{0.3cm} $T \gets \emptyset$ \\
\hspace{1cm} (2) \hspace*{0.3cm} $M_{Y} \gets \emptyset$ \quad \textit{*(this variable accumulates plank edges below)*} \\
\hspace{1cm} (3) \hspace*{0.3cm} 
$M_{N} \gets M_{P}$ \quad \textit{*(variable for non-plank edges)*} \\
\hspace{1cm} (4) \hspace*{0.3cm} \textbf{while} there is a plank edge $e \equiv (u,v) \in M_{N}$\\
\hspace{1cm} (5) \hspace*{1.1cm} 
$M_{Y} \gets M_{Y} \cup \{e\}$ \\
\hspace{1cm} (6) \hspace*{1.1cm} 
$M_{N} \gets M_{N} \setminus \{e\}$ \\
\hspace{1cm} (7) \hspace*{1.1cm} 
$T \gets T \cup \{e_{u},e_{v}\}$ \\
\hspace{1cm} (8) \hspace*{1.1cm} 
$M_{G} \gets M_{G} \cup \{e_{u},e_{v}\}$ \\
\hspace{1cm} (9) \hspace*{0.3cm} 
$M_{G} \gets M_{G} \cup M_{N}$\\
\vspace*{-3mm}
\end{mdframed}
\vspace*{-4mm}
\caption{}
\label{fig:Procedure_4}
\end{figure}
\end{minipage}

    Note that the above procedure adds one edge in $M_{Y}$ and two edges in $T$ in every iteration of the while loop. Therefore, $|T| = 2 \cdot |M_{Y}|$. Also, note that $T \subseteq M_{G}$.
    Now, we complete the construction of the vertex cover $VC_{G}$. We consider two sub-cases based on the size of $M_{Y}$. And, for each of the sub-cases, we construct the vertex cover separately.
    \begin{enumerate}
        \item \uline{Sub-case:} ($|M_{Y}| \geq \delta k$)
        
        For every edge in $M_{P}$, we simply add both its endpoints to $VC_{G}$. 
        This completes the construction of $VC_{G}$. It covers all edges of $G_{R}$ since all edges are incident on some edge of $M_{P}$. Let us compute the size of $VC_{G}$.
        Note that for every edge in $M_{Y} \subseteq M_{P}$, there are two edges in $T \subseteq M_{G}$. Therefore, the size of vertex cover is:
        \[
        |VC_{G}| = 2 |M_{G}| - |T| =2 |M_{G}| - 2 |M_{Y}| \leq 2|M_{G}| - 2\delta k \leq 2k - 2\delta k
        \]

        \item \uline{Sub-case:} ($|M_{Y}| < \delta k$)
        
        For this sub-case, we construct the vertex cover in the following manner. For every edge in $T$, we add its both endpoints to $VC_{G}$. And, we remove all the edges in $G_{R}$ covered by them. The remaining graph contains the set $M_{N}$, and some blue edges incident on it. Since $M_{N}$ is defined on non-plank edges, the remaining blue edges can not incident on both endpoints of any edge of $M_{N}$. Now, for every edge in $M_{N}$, we pick its that endpoint in the vertex cover that it shares with the blue edges incident on it. It completes the construction of $VC_{G}$, and it covers all edges of $G_{R}$. Note that for every edge in $M_{G}$, we added its both endpoints to $VC_{G}$ except the edges that came from $M_{N}$ for which, we just added one endpoint in $VC_{G}$. Also, note that $|M_{N}| = |M_{P}| - |M_{Y}| > 3 \delta k - \delta k = 2 \delta k$.
        Therefore, the size of the vertex cover is:
        \[
        |VC_{G}| = 2|M_{G}| - |M_{N}| < 2 |M_{G}| - 2\delta k \leq 2k - 2\delta k
        \]

    \end{enumerate}
    
\end{itemize}
Hence, we have a vertex cover of size at most $2k - 2\delta k$. This completes the proof of the lemma.
\end{proof}

\noindent Based on the above lemma, we will assume that all single edge clusters can be covered with $(\frac{2t_{1}'}{3} + 8 \delta k) \leq (\frac{2t_{1}}{3} + 8 \delta k)$ vertices; otherwise the graph has a vertex cover of size at most $(2k - 2\delta k)$ and the soundness proof would be complete. 
Now, we bound the vertex cover of the entire graph in the following manner.
\begin{align*}
|VC(G)| &\leq |VC\left( \cup_{i=1}^{t_1'}P_{i}\right)| + |VC'| \\
&= |VC\left( \cup_{i=1}^{t_1'}P_{i}\right)| + |VC\left( (\cup_{j=1}^{t_2} S_j) \cup (\cup_{k=1}^{t_3} W_k) \cup (\cup_{l=1}^{t_4} Y_l)\right)|\\
&\leq  \sum_{i = 1}^{t_{1}'}|VC(P_{i})| + \sum_{i = 1}^{t_{2}}|VC(S_{i})| + \sum_{i = 1}^{t_{3}}|VC(W_{i})| + \sum_{i = 1}^{t_{4}}|VC(Y_{i})|\\
&\leq \left( \frac{2t_{1}}{3} + 8\delta k \right) + t_{2} + \sum_{i = 1}^{t_{3}} \left( \del[W_{i}] + 1.62  \right) + \sum_{i = 1}^{t_{4}} \left( \del[Y_{i}] + 1.8  \right), \\
& \hspace{9cm} \textrm{(using Lemmas~\ref{lemma:non_star_vertex_cover_1},~\ref{lemma:non_star_vertex_cover_2}, and~\ref{lemma:single_star_vertex_cover})} \\
&= (0.67)t_{1} + 8\delta k + t_{2} + (1.62)t_{3} + (1.8) t_{4}  + \left( \sqrt{2}+1 \right) \left( \sum_{i = 1}^{t_{3}} \delta(W_{i}) + \sum_{i = 1}^{t_{4}} \delta(Y_{i})
\right)
\end{align*}

\noindent Since the optimal cost $OPT (\X,k) = \mathlarger{\sum}_{j = 1}^{k} \sqrt{m_{j}(m_{j}-1)} + \mathlarger{\sum}_{i=1}^{t_{3}}  \delta(W_{i}) + \mathlarger{\sum}_{i=1}^{t_{4}}  \delta(Y_{i}) \leq m-k/2+\delta k$, we get $\mathlarger{\sum}_{i=1}^{t_{3}}  \delta(W_{i}) + \mathlarger{\sum}_{i=1}^{t_{4}}  \delta(Y_{i}) \leq m-k/2+\delta k - \mathlarger{\sum}_{j = 1}^{k} \sqrt{m_{j}(m_{j}-1)}$. We substitute this value in the previous equation, and get the following inequality:

\begin{align*}
|VC(G)| &\leq  (0.67)t_{1} + 8\delta k + t_{2} + (1.62) t_{3} +(1.8) t_{4} + \left( \sqrt{2}+1 \right) \cdot \left( m-k/2 -\sum_{j=1}^{k} \sqrt{m_{j}(m_{j}-1)} + \delta k \right)
\end{align*}

\noindent Using Lemma~\ref{lemma:basic_bound}, we obtain the following inequalities:
\begin{enumerate}
    \item For $P_{j}$, $\st[|P_{j}|] \geq |P_{j}| - 1$ since $|P_{j}| = 1$
    \item For $S_{j}$, $\st[|S_{j}|] \geq |S_{j}| - (2 - \sqrt{2})$ since $|S_{j}| \geq 2$
    \item For $W_{j}$, $\st[|W_{j}|] \geq |W_{j}| - (2 - \sqrt{2})$ since $|W_{j}| \geq 2$
    \item For $Y_{j}$, $\st[|Y_{j}|] \geq |Y_{j}| - (3 - \sqrt{6})$ since $|Y_{j}| \geq 3$
\end{enumerate} 
We substitute these values in the previous equation, and get the following inequality:
\begin{align*}
|VC(G)| &\leq  (0.67)t_{1} + 8\delta k + t_{2} + (1.62) t_{3} +(1.8) t_{4}  + \left( \sqrt{2}+1 \right) \cdot \left( m-k/2  - \sum_{j=1}^{t_{1}} \left( |P_{j}| - 1\right) + \right.\\
&  \quad \left.  - \sum_{j=1}^{t_{2}} \left( |S_{j}| - (2-\sqrt{2})\right) -  \sum_{j=1}^{t_{3}} \left( |W_{j}| - (2-\sqrt{2})\right) - \sum_{j=1}^{t_{4}} \left( |Y_{j}| - (3-\sqrt{6})\right)  + \delta k \right)
\end{align*}

\noindent Since the number of edges $|G| = \mathlarger{\sum}_{j = 1}^{t_{1}} \, \, |P_{j}| + \mathlarger{\sum}_{j = 1}^{t_{2}} \, \, |S_{j}| + \mathlarger{\sum}_{j = 1}^{t_{3}} \, \, |W_{j}|  + \mathlarger{\sum}_{j = 1}^{t_{4}} \, \, |Y_{j}| $, we get the following inequality:
\begin{align*}
|VC(G)| &\leq  (0.67)t_{1} + 8\delta k + t_{2} + (1.62) t_{3} +(1.8) t_{4}  + \left( \sqrt{2}+1 \right) \cdot \Bigg( -k/2 + t_{1} + t_{2} \cdot \left( 2-\sqrt{2} \right) + \Bigg. \\
& \hspace{5cm} \Bigg. + \,\, t_{3} \cdot \left( 2-\sqrt{2} \right) + t_{4} \cdot \left( 3-\sqrt{6} \right) + \delta k \Bigg)
\end{align*}
\noindent We substitute $k = t_{1} + t_{2} + t_{3} + t_{4}$, and obtain the following inequality:
\begin{align*}
|VC(G)| &\leq (0.67)t_{1} + 8\delta k + t_{2} + (1.62) t_{3} +(1.8) t_{4}  + \left( \sqrt{2}+1 \right) \cdot \left( \frac{t_{1}}{2} + \frac{t_{2}}{10} + \frac{t_{3}}{10} + \frac{3t_{4}}{50}  +  \delta k \right) \\ 
&=  (1.88)t_{1} + (1.25)t_{2} + (1.87) t_{3} + (1.95) t_{4} + \left( \sqrt{2}+9 \right) \delta k \\
&< (1.95) k + \left( \sqrt{2}+9 \right) \delta k \hspace{3cm} \textrm{(using $t_{3} + t_{4} + t_{1} + t_{2} = k $)}\\
&\leq (2-\veps)k, \quad \textrm{for appropriately small constants $\veps,\delta>0$}
\end{align*}

\noindent This proves the soundness condition and it completes the proof of Theorem~\ref{theorem:reduction_kmedian}.
In the remaining part, we prove Lemmas~\ref{lemma:non_star_vertex_cover_1} and~\ref{lemma:non_star_vertex_cover_2}.

\section{Vertex Cover of Non-Star Graphs}\label{section:vertex_cover}
In this section, we bound the vertex cover size of any non-star graph $F$. 
We aim to obtain this bound in terms of $\delta(F)$, i.e., the extra cost of the graph $F$. 
To do so, we require a bound on the extra cost. 
The $1$-median cost of an arbitrary non-star graph is tricky to compute. 
Fortunately, we do not require the exact optimal cost of a graph; a lower bound on the optimal cost suffices. 
Furthermore, for some graph instances we can compute the exact optimal cost. 
For example, a star graph with $r$ edges corresponds to a regular simplex of size length $s$ that has the optimal $1$-median cost: $s \cdot \sqrt{\frac{r(r-1)}{2}}$. We proved this earlier in Lemma~\ref{lemma:simplex_cost}. 
For the more complex graph instances, we use the following decomposition lemma to bound their optimal cost.

\begin{lemma}[Decomposition lemma]\label{lemma:decomposition}
Let $G = (V, E)$ be any graph and let $E_1, ..., E_t$ be any partition of edges and let $G_{1}, G_{2}, ..., G_{t}$ be the subgraphs induced by these edges respectively.
The following inequality bounds the optimal cost of $G$ in terms of the optimal costs of subgraphs $G_1, ..., G_t$.
\[
\Phi^{*}(G) \geq \Phi^{*}(G_{1}) + \dotsc + \Phi^{*}(G_{t})
\]
\end{lemma}
\begin{proof}
Let $f^{*}$ be an optimal median of $\X(G)$. For any edge $e \in E$, let $x_{e}$ denote the corresponding point in $\X(G)$.
The proof follows from the following sequence of inequalities:
$
    \Phi^{*}(G) = \sum_{e \in E} d(x_e,f^{*}) 
    = \sum_{i = 1}^{t} \sum_{e \in E_{i}} d(x_e,f^{*}) 
    \geq \sum_{i = 1}^{t} \Phi^{*}(G_{i})
$.
\end{proof}

\noindent
If we can compute the optimal cost of each subgraph, then we can lower bound the overall cost of the graph using the above decomposition lemma. 
In general, a graph may be decomposed in various ways.
However, we prefer that decomposition that gives better bound on the optimal cost.
In the next subsection, we use this decomposition lemma to bound the extra cost of any non-star graph. 
We will use that bound in subsequent subsections to prove Lemmas~\ref{lemma:non_star_vertex_cover_1} and~\ref{lemma:non_star_vertex_cover_2}.

\subsection{$\mathbf{1}$-Median Cost of Non-Star Graphs}
In this section, we show that any non-star graph on $m$ edges has the optimal $1$-median cost at least $m - 0.342$ which is at least $\sqrt{m(m-1)} + 0.158$. 
We assume all graphs to be triangle-free since hard \vc instances are triangle-free. 
Therefore, we do not explicitly mention the ``triangle-freeness'' of graphs whenever we state a lemma related to graphs. 

To obtain a lower bound on the optimal cost, we decompose a graph into so called ``Fundamental non-star graphs''. 
We will show this decomposition process later. 
For now, we first bound the $1$-median cost of these fundamental graphs. 
We then bound the cost of any graph using the decomposition lemma. 
Here is the formal description of fundamental graphs.

\begin{definition}[Fundamental Non-Star Graph]
A fundamental non-star graph is a graph that becomes a star graph when \textbf{any} pair of vertex-disjoint edges are removed from it. The graph that has only two vertex-disjoint edges is also a fundamental non-star graph.
\end{definition}

\begin{figure}[H]
    \centering
    \begin{framed}
    \includegraphics{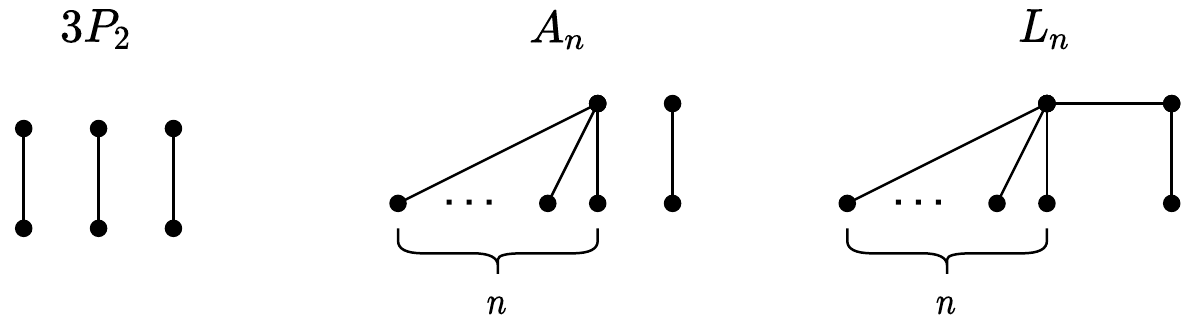}
    \end{framed}
    \vspace*{-3mm}
    \caption{Fundamental non-star graphs: $3$-$P_{2}$, $A_{n}$, and $L_{n}$.}
    \label{fig:graphs}
\end{figure}

\noindent
The following lemma shows that there are exactly three types of fundamental non-star graphs. These are shown in Figure~\ref{fig:graphs}.

\begin{lemma}\label{lemma:fundamental_graphs}
There are only three types of fundamental non-star graphs: $3$-$P_{2}$, $A_{n}$, and $L_{n}$, for $n \geq 1$. These are shown in Figure~\ref{fig:graphs}.
\end{lemma}
\begin{proof}
It is easy to see that the three types of graphs $3$-$P_2$, $A_n$, and $L_n$ shown in Figure~\ref{fig:graphs} are fundamental non-star graphs. We need to argue that these are the {\em only} fundamental non-star graphs. For this we do a case analysis.
Let $M$ denote a maximum matching of any graph $G = (V, E)$.
We divide the analysis into three cases based on the size of the maximum matching.

\begin{enumerate}
    \item \uline{Case 1:} $|M| \geq 4$. 

Suppose we remove any pair of vertex-disjoint edges from $M$. Then, the remaining graph is still a non-star graph due to a matching of size at least two. Therefore, any such graph can not be a fundamental non-star graph.

\item \uline{Case 2:} $|M| = 3$. 

Let $U$ denote the set of unmatched edges in the graph, i.e, $U = E \setminus M$. Suppose $U \neq \emptyset$, and $l$ be any edge in $U$. Observe that $l$ can incident on at most two  edges of $M$. Since $M$ has a size exactly three, there is always an edge in $M$ that is vertex-disjoint from $l$. Let this edge be $m_{e} \in M$. The set of edges $\{l,m_{e}\}$ forms a $2$-$P_{2}$ (i.e., two vertex-disjoint edges), and we can remove it from $G$. The remaining graph still has a matching of size at least two. Therefore, such a graph can not be a fundamental non-star graph. Now, let us consider the case when $U = \emptyset$. In this case, the graph is equivalent to a $3$-$P_{2}$. 

\item \uline{Case 3:} $|M| = 2$. 
 
 Let $e_{1} \coloneqq (u_{1},v_{1})$, and $e_{2} \coloneqq (u_{2},v_{2})$ be the matching edges in $M$. Let $U = E \setminus M$ denote the set of unmatched edges. If $U$ forms a non-star graph, we can remove a pair of vertex-disjoint edges from it. The remaining graph is a non-star due to $|M| = 2$. Therefore, any graph with $|M| = 2$ and $U$ as a non-star graph, can not be a fundamental non-star graph. Therefore, let us consider the case when $U$ forms a star graph. Suppose all edges of $U$ share a common vertex $w$. There are two possibilities: $w$ is an endpoint of some matching edge or not. Let us consider these two possibilities one by one.

 \begin{enumerate} [label=(\alph*)]
     \item \uline{Subcase:} $w$ is an endpoint of some matching edge
     
 Without loss of generality, we can assume $w \equiv u_{1}$. 
 Now, no two edges of $U$ can be incident on $u_{2}$ and $v_{2}$; otherwise, it would form a triangle $(w,u_{2},v_{2})$. 
 Therefore, at most one edge of $U$ can incident on either $u_{2}$ or $v_{2}$. 
 If such an edge exist, without loss of generality, we can assume that it is incident on $u_{2}$. 
 This graph is of type $L_{n}$. 
 On the other hand, if no edge of $U$ is incident on either $u_{2}$ or $v_{2}$, the graph is of type $A_{n}$. 

  \item \uline{Subcase:} $w$ is not an endpoint of any matching edge.
  
If $|U| = 1$, the graph is simply $A_{1}$.
Now, let us assume that $|U| \geq 2$. Let
$l_{1}$ and $l_{2}$ be any two edges in $U$. The edges: $l_{1}$ and $l_{2}$, can not incident on the same matching edge say $(u_{1},v_{1})$; otherwise it forms a triangle: $(w,u_{1},v_{1})$. Also, note that every edge $e \in U$ must be incident on some matching edge; otherwise $M \cup \{e \}$ would form a matching of size $> |M|$. It would contradict that $M$ is a maximum matching. Therefore, without loss of generality, we can assume that $l_{1}$ is incident on $u_{1}$, and $l_{2}$ is incident on $u_{2}$. Now, it forms a path of length four: $(v_{1},u_{1},w,u_{2},v_{2})$. We can always remove an alternating pair of edges from the path, and the resulting graph still contains a pair of vertex-disjoint edges, i.e., a $2$-$P_{2}$. Therefore, any such graph can not be a fundamental non-star graph.
\end{enumerate}
\end{enumerate}
From the above case analysis we conclude that all fundamental non-star graphs are of type $3$-$P_2$, $A_n$, or $L_n$.
This completes the proof of the lemma.
\end{proof}

Next, we bound the optimal $1$-median cost of each fundamental non-star graph. 
In this discussion, we will use $r$ to denote the number of edges in various cases.

\begin{lemma}\label{lemma:cost_3P2}
Let $r \coloneqq 3$ denote the number of edges in $3$-$P_2$. The optimal cost of $3$-$P_{2}$ is at least $(r + 0.46)$.
\end{lemma}
\begin{proof}
It is easy to see that $\X(3$-$P_{2})$ forms a simplex of side length 2. Therefore, we can use Corollary~\ref{corollary:cost_non_star} to compute its optimal $1$-median cost:
$
\Phi^{*}(3\textrm{-}P_{2}) = \sqrt{2} \sqrt{3(3-1)} = 2\sqrt{3} \geq r + 0.46 \quad \textrm{(since $r = 3$)}.
$
\end{proof}

\begin{lemma}\label{lemma:cost_An}
Let $r \coloneqq n+1$ denote the number of edges in $A_n$. The optimal cost of $A_{n}$ is at least:
\begin{enumerate}
    \item $r$, for $n \geq 1$.
    \item $r + 0.095$, for $n \geq 2$.
    \item $r + 0.135$, for $n \geq 3$.
\end{enumerate}
\end{lemma}

\begin{proof}
Consider the point set $\X(A_{n})$. 
It forms a simplex with $(r-1)$ points at a distance of $\sqrt{2}$ from each other and the remaining point at a distance of $2$ from the rest of the points.
Based on Fact~\ref{fact:isometric}, we represent the coordinates of every point in an $(r-1)$-dimensional space in the following way.
\[
    a_{1} = (1,0,...,0), \quad a_{2} = (0,1,...,0), \quad ..., \quad a_{r-1} = (0,...,0,1), \quad a_{r}= (u,\,\dotsc, \,u)
\]
Here $u = \frac{1}{r-1} + \sqrt{\frac{3}{r-1} + \frac{1}{(r-1)^2}}$.
Let $(c_{1},\dotsc,c_{r-1})$ be an optimal $1$-median of $S = \{a_1, ..., a_r\}$. 
If $c_{i} \neq c_{j}$ for any $i \neq j$, we can swap $c_{i}$ and $c_{j}$ to create a different median with the same $1$-median cost. 
Since the $1$-median is always unique for a set of non-collinear points (by Fact~\ref{fact:unique_median}), we assume $c^{*} = \{c,\dotsc,c\}$ as the optimal median. Then, the optimal $1$-median cost is:
\begin{align*}
\Phi(c^{*},S) = \Phi(c^{*},a_{r}) + \sum_{i = 1}^{r-1} \Phi(c^{*},a_{i}) 
= (u-c)\cdot \sqrt{r-1} + (r-1) \cdot \sqrt{ (1-c)^{2} + (r-2) c^{2}}
\end{align*}
The function is strictly convex and attains minimum at $c = \fdfrac{\sqrt{r}+1}{\sqrt{r} \cdot (r-1)}$. 
We get the following optimal cost on substituting the values of $t$ and $c$ in the previous equation: 
\begin{align*}
\Phi(c^{*},S) &= \sqrt{r(r-1)} + \sqrt{3+\frac{1}{r-1}} - \sqrt{\frac{r}{r-1}} = \sqrt{r(r-1)} + \frac{2}{ \sqrt{3+\frac{1}{r-1}} + \sqrt{\frac{r}{r-1}}}
\end{align*}
For $r \geq t$, we get: 
\begin{align*}
\Phi(c^{*},S)& \geq \sqrt{r(r-1)} + \frac{2}{ \sqrt{3+\frac{1}{t-1}} + \sqrt{\frac{t}{t-1}}}, \quad \quad \left(\textrm{$\because  r \geq t$}\right) \\
&= \sqrt{r(r-1)} + \sqrt{3+\frac{1}{t-1}} - \sqrt{\frac{t}{t-1}} \\
&\geq r - (t-\sqrt{t(t-1)}) + \sqrt{3+\frac{1}{t-1}} - \sqrt{\frac{t}{t-1}}, \quad \quad \textrm{(using Lemma~\ref{lemma:basic_bound})}
\end{align*}

\noindent Substituting $t = 2$, we get:
$
\Phi(c^{*},S) \geq r - (2 -\sqrt{2}) + \sqrt{4} - \sqrt{2} \geq r.
$

\noindent Substituting $t = 3$, we get:
$
\Phi(c^{*},S) \geq r - (3 -\sqrt{6}) + \sqrt{7/2} - \sqrt{3/2} > r + 0.095.
$

\noindent Substituting $t = 4$, we get:
$\Phi(c^{*},S) \geq r - (4 -\sqrt{12}) + \sqrt{10/3} - \sqrt{4/3} > r + 0.135.$\\
This completes the proof of the lemma.
\end{proof}

\begin{lemma}\label{lemma:cost_Ln}
Let $r \coloneqq n+2$ denote the number of edges in $L_n$.
Then the optimal cost of $L_{n}$ is at least:
\begin{enumerate}
    \item $r - 0.268$, for $n = 1$.
    \item $r - 0.334$, for $n = 2$.
    \item $r - 0.342$, for $n \geq 3$.
\end{enumerate}
\end{lemma}

\begin{proof}
Let us prove the first statement for $r=3$ which corresponds to the graph being $L_1$ (i.e., $n=1$).
In $\X(L_1)$, there are two points at distance of 2 from each other, and the third point at a distance of $\sqrt{2}$ from the other two points. It forms a simplex $S$ of dimension two. Based on Fact~\ref{fact:isometric}, we represent the coordinates of the simplex in the following way:
\[
    a_{1} = (0,0), \quad 
    a_{2} = (\sqrt{2},0), \quad
    a_{3} = ( 0, \sqrt{2}).
\]

\noindent Note that the pairwise distances are preserved in this representation.
Let $(c_{1},c_{2})$ be the optimal $1$-median of $S$. 
If $c_{i} \neq c_{j}$ for any $i \neq j$, we can swap $c_{i}$ and $c_{j}$ to create a different median with the same 1-median cost. Therefore, we consider $c^{*} \coloneqq (c,c)$ as the optimal $1$-median of $S$. Then, we get the following optimal $1$-median cost of $S$.
\begin{align*}
\Phi(c^{*},S) = \Phi(c^{*},a_{1}) + \Phi(c^{*},a_{2}) +\Phi(c^{*},a_{3}) = \sqrt{2} \cdot c + 2 \cdot \sqrt{\left( c-\sqrt{2} \right)^{2} + c^{2}} 
\end{align*}
The function is strictly convex and attains minimum at $c = \sqrt{\frac{1}{2}} - \sqrt{\frac{1}{6}}$. Substituting the value of $c$ in $\Phi(c^{*},S)$, we get $\Phi(c^{*},S) = 1+\sqrt{3} > r - 0.268$, for $r = 3$. This completes the proof of the first statement.

Let us prove the second statement. 
Here, we have $r = 4$ (or $n=2$). 
We create three copies of $L_{2}$ (i.e., $3$-$L_{2}$), and decompose them into three subgraphs: $2$-$P_{4}$, $A_{2}$, and $S_{3}$. 
The decomposition is shown in Figure~\ref{fig:graph_l2}. Note that $P_{4}$ is the same as $L_{1}$. 
There are also other ways to decompose the graph. 
However, some of those decompositions give weak bound on the optimal cost. And, this decomposition gives sufficiently good bound on the optimal cost of $L_{2}$.

\begin{figure}[H]
    \centering
    \begin{framed}
    \includegraphics[scale = 0.65]{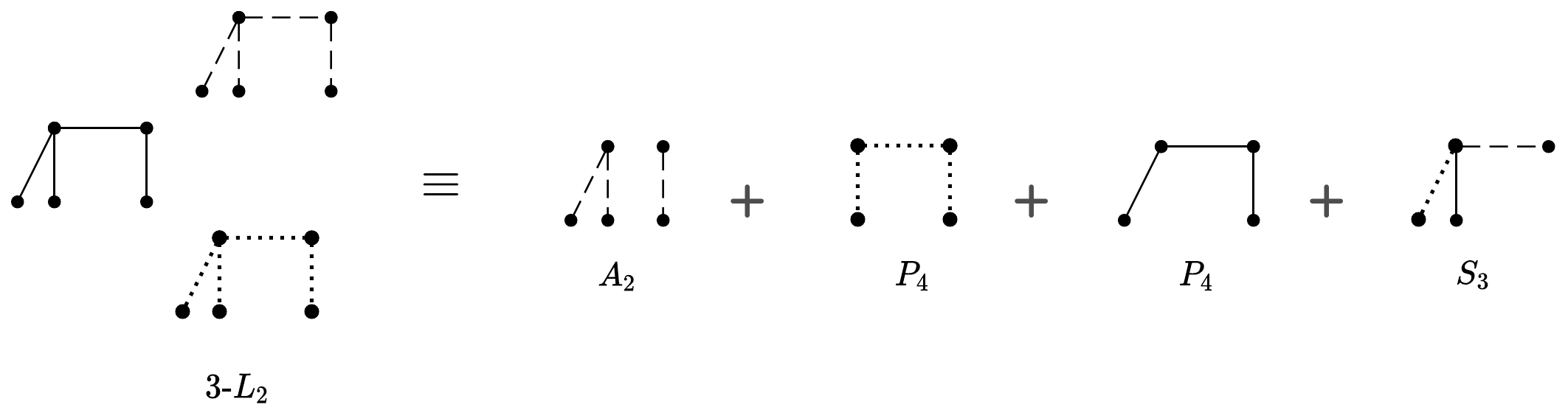}
    \end{framed}
    \vspace*{-3mm}
    \caption{Decomposition of $3$-$L_{2}$}
    \label{fig:graph_l2}
\end{figure}

\noindent Let $c^{*}$ be the optimal 1-median for $L_{2}$. Based on the decomposition, we bound the optimal 1-median cost of $L_{2}$ as follows:
\begin{align*}
    3 \cdot \Phi^{*}(L_{2}) &\geq \Phi^{*}(A_{2}) + 2 \cdot\Phi^{*}(P_{4}) + \Phi^{*}(S_{3}) \numberthis \label{eqn:1}
\end{align*}

\noindent We already have the bounds on the optimal costs of $A_{2}$, $P_{4}$, and $S_{3}$. That is,
\begin{itemize}
\item For $A_{2}$, we have $\Phi^{*}(A_{2}) \geq 3 + 0.095$. This follows from Statement 2 of Lemma~\ref{lemma:cost_An}.

\item For $P_{4}$, we have $\Phi^{*}(P_{4}) \geq |P_{4}| - 0.268 = 2.732$. This follows from Statement 1 of  Lemma~\ref{lemma:cost_Ln} since $P_{4}$ is the same as $L_{1}$, and that the number of edges in $P_4$, denoted by $|P_4|$, equals $3$. 

\item For $S_{3}$, we have $\Phi^{*}(S_{3}) = \sqrt{3(3-1)} = \sqrt{6}$. This follows from Corollary~\ref{corollary:cost_star}, for $r = 3$.
\end{itemize}

Substituting the above values in Equation (\ref{eqn:1}), we get the following inequality:
\begin{align*}
    3 \cdot \Phi^{*}(L_{2}) \geq 3.095 + 2 \cdot (2.732) + \sqrt{6} > 11 
\end{align*}
Thus, the optimal cost of $L_{2}$ is at least $11/3 \geq 3.666 = |L_{2}| - 0.334$. This completes the proof of Statement 2.

Now, let us prove the third statement. Here, we have $r \geq 5$ (or $n \geq 3$). We create two copies of $L_{n}$, and decompose it into three subgraphs: $A_{n}$, $S_{n}$, and $P_{4}$. The decomposition is shown in Figure~\ref{fig:graph_ln}. Again, note that there are many ways to decompose the graph. However, those decompositions may yield weak bound on the optimal cost. Whereas, this decomposition gives sufficiently good bound on the optimal $1$-median cost of $L_{n}$.
\begin{figure}[H]
    \centering
    \begin{framed}
    \includegraphics[scale = 0.7]{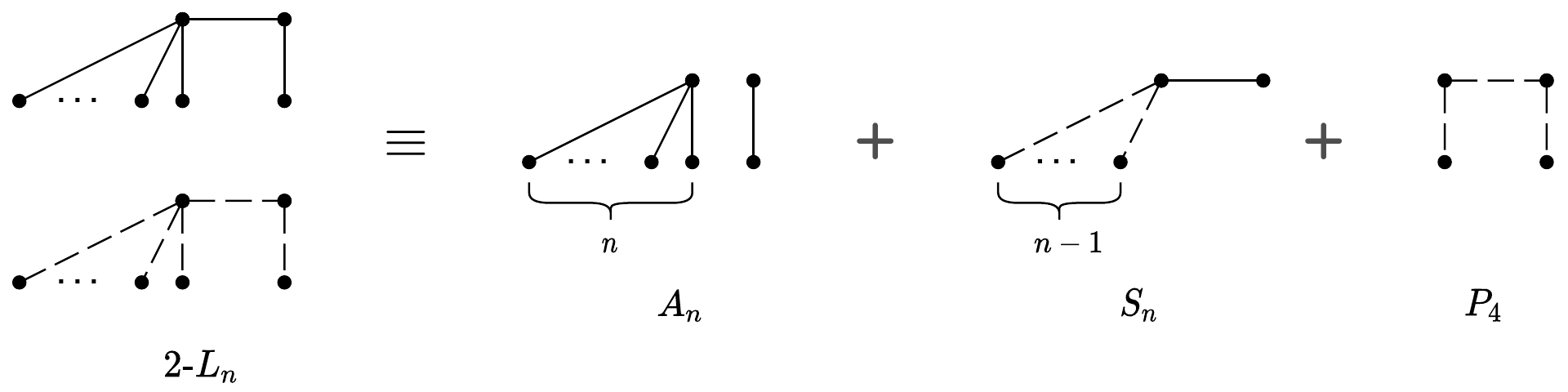}
    \end{framed}
    \vspace*{-3mm}
    \caption{Decomposition of $2$-$L_{n}$ for $n \geq 3$}
    \label{fig:graph_ln}
\end{figure}
\noindent Let $c^{*}$ be the optimal 1-median for $L_{n}$. Based on the decomposition, we bound its optimal $1$-median cost in the following manner:
\begin{align*}
    2 \cdot \Phi^{*}(L_{n}) &\geq \Phi^{*}(A_{n}) + \Phi^{*}(S_{n}) +  \Phi^{*}(P_{4}) \numberthis \label{eqn:2} 
\end{align*}

\noindent We already know the bounds on the optimal costs of $A_{n}$, $S_{n}$ and $P_{4}$. That is,

\begin{itemize}
\item For $A_{n}$, we have $\Phi^{*}(A_{n}) \geq (n+1) + 0.135$. This follows from Statement 3 of Lemma~\ref{lemma:cost_An} ($\because$ $n \geq 3$).

\item For $S_{n}$, we have $\Phi^{*}(S_{n}) = \sqrt{n(n-1)} \geq n - (3-\sqrt{6})$. The first equality follows from Corollary~\ref{corollary:cost_star}) whereas the second inequality follows from Lemma~\ref{lemma:basic_bound}, for $n \geq 3$.

\item  Since $P_{4}$ is is the same as $L_{1}$, we have $\Phi^{*}(P_{4}) \geq |P_{4}| - 0.268$. This follows from Statement 1 of Lemma~\ref{lemma:cost_Ln}.
\end{itemize}

\noindent Substituting the above values in Equation (\ref{eqn:2}) and using the fact that $|P_4|=3$, we obtain the following inequality:
\begin{align*}
    2\cdot \Phi^{*}(L_{n}) \geq (2n+1+|P_{4}|) + (0.135 + \sqrt{6} - 3 -0.268) \geq 2 \cdot (n+2) - 0.684
\end{align*}
\noindent Thus, the optimal cost of $L_{n}$ is at least $(n+2) - 0.342 = |L_{n}| - 0.342$. This completes the proof of the lemma.
\end{proof}

\begin{corollary} \label{corollary:cost_fundamental_graphs}
The cost of any fundamental non-star graph with $r$ edges is at least $r-0.342$.
\end{corollary}
\begin{proof}
The proof simply follows from Lemmas~\ref{lemma:cost_3P2},~\ref{lemma:cost_An}, and~\ref{lemma:cost_Ln}.
\end{proof}

We will now bound the cost of {\em any} non-star graph by decomposing it into fundamental non-star graphs.
For this, we define the concept of ``\emph{safe pair}''. 
A safe pair is a pair of vertex-disjoint edges in the graph such that when we remove it from the graph, the remaining graph remains a non-star graph. 
Let us see why such a pair is important. 
First, note that the optimal cost of a $2$-$P_{2}$ is exactly two, using Corollary~\ref{corollary:cost_non_star}. 
Suppose we remove a $2$-$P_{2}$ from the graph $F$. 
Let the remaining graph be $F'$. 
Suppose $\Phi^{*}(F') = |F'| + \gamma$, where $\gamma$ is some constant. 
Then it is easy to see that $\Phi^{*}(F) \geq \Phi^{*}(F') + \Phi^{*}(2$-$P_{2}) = |F| + \gamma$. 
Note that $\gamma$ value is preserved in this decomposition. Suppose we keep removing a safe pair from $F$ until we obtain a graph that does not contain any safe pair. 
Then the remaining graph is simply a fundamental non-star graph by the definition of fundamental non-star graph. 
Moreover, we showed earlier that the optimal cost of any fundamental non-star graph with $r$ edges, is at least $r - 0.342$. 
Note that here $\gamma \geq -0.342$. 
Therefore, $F$ has the optimal cost at least $|F| - 0.342$. This is the main reason for defining a safe pair and a fundamental non-star graphs.
The decomposition procedure described above is given below.

\begin{figure}[H]
\begin{mdframed}
{\tt Decompose$(F)$} \vspace*{1mm} \\
\hspace*{1.2cm} {\bf Input}: A non-star graph $F$.\\
\hspace*{1.2cm} {\bf Output}: A fundamental non-star graph $D \in \{ 3 \textrm{-} P_{2}, A_{n}, L_{n}\}$ \vspace*{1mm}\\
\hspace{1cm} (1) \hspace*{0.6cm} $D \gets F$ \\
\hspace{1cm} (2) \hspace*{0.6cm} while $D \notin \{ 3 \textrm{-} P_{2}, A_{n}, L_{n}\}$ \\
\hspace{1cm} (3) \hspace*{1.5cm} Let $\{e,e'\} \in E(D)$ be a safe pair  \\
\hspace{1cm} (3) \hspace*{1.5cm} $E(D) \gets E(D) \setminus \{e,e'\}$ \\
\hspace{1cm} (4) \hspace*{0.6cm} return $D$
\end{mdframed}
\vspace*{-3mm}
\caption{Decomposition of any non-star graph $F$ into the \emph{fundamental} non-star graphs.}
\label{fig:decompose_procedure}
\end{figure}

\noindent
The above discussion is formalised as the next lemma.

\begin{lemma}\label{lemma:decompose_lemma}
The cost of any non-star graph $F$ is at least $(|F| -0.342) \geq \sqrt{|F|(|F|-1)}+0.158$.
\end{lemma}

\begin{proof}
Suppose the procedure {\tt Decompose}($F$) runs the while loop $t$ times. 
This means that $F$ is composed of $t$ safe pairs and a fundamental non-star graph $D \in \{ 3 \textrm{-} P_{2}, A_{n}, L_{n}\}$. 
We call $D$ the \emph{residual} graph of $F$. 
Note that $D$ has exactly $|F|-2t$ edges. 
Also, note that $2$-$P_{2}$ is a fundamental non-star graph since it the same as $A_{1}$. 
Based on the decomposition of $F$ into $2$-$P_{2}$'s and $D$, we bound the optimal cost of $F$ as follows: 
\begin{align*}
    \Phi^{*}(F) &\geq t \cdot \Phi^{*}(2\textrm{-}P_{2}) + \Phi^{*}(D) &&\\
    &= 2 \cdot t + \Phi^{*}(D) \quad && \textrm{(using Corollary~\ref{corollary:cost_non_star})}\\
    &\geq 2 \cdot t + (|F|-2t) - 0.342 && \textrm{(using Corollary~\ref{corollary:cost_fundamental_graphs})}\\
    &= |F| - 0.342 &&\\
    &> \sqrt{|F|(|F|-1)}+1/2-0.342  &&\textrm{(using Lemma~\ref{lemma:basic_bound})}\\
    &= \sqrt{|F|(|F|-1)} + 0.158 &&
\end{align*}
This completes the proof of the lemma.
\end{proof}

Next, we show a stronger bound on the optimal cost than the one stated in the previous lemma. 
However, this bound applies for a particular set of graph instances. 
For positive integers $p, q$, let us define a new graph $L_{p,q}$.
This graph is composed of two star graphs $S_{p}$ and $S_{q}$, such that there is an edge between the center vertices of $S_{p}$ and $S_{q}$. Here, the {\em center} vertex is the vertex that is the common endpoint of all edges in a star graph and the remaining vertices are called {\em pendent} vertices.
Let $s_{1}$ and $s_{2}$ denote the center vertices of $S_{p}$ and $S_{q}$ respectively. 
We call the edge $(s_1,s_2)$, the \emph{bridge} edge, and the graph $L_{p,q}$, the \emph{bridge} graph. 
Also, we call the pendent vertices of $S_{p}$ and $S_{q}$ as the \emph{left} and \emph{right} pendent vertices respectively.
Please see Figure~\ref{fig:updated_Ln} for the pictorial depiction of $L_{p,q}$. 
Note that when $p = n$ and $q = 1$, the bridge graph is the same as $L_{n}$.

\begin{figure}[H]
    \centering
    \begin{framed}
    \vspace{-0.7cm}
    \includegraphics{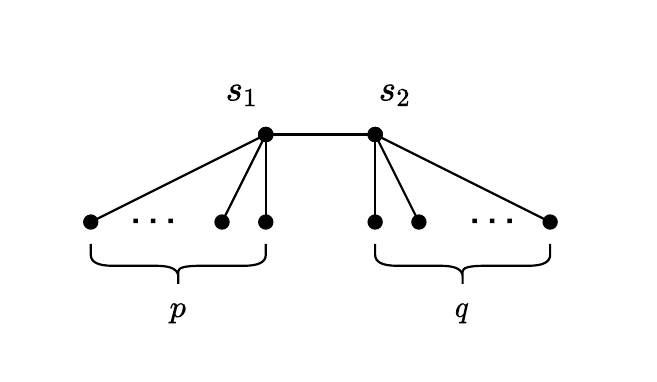}
    \vspace{-0.5cm}
    \end{framed}
    \vspace{-0.3cm}
    \caption{A Bridge Graph: $L_{p,q}$, for $p,q \geq 1$.}
    \label{fig:updated_Ln}
\end{figure}

\noindent Now, we state the lemma that bounds the optimal $1$-median cost of any non-star non-bridge graph.
\begin{lemma}\label{lemma:ultra_decompose}
Suppose $F$ is a non-star non-bridge graph. 
Then $F$ has the optimal $1$-median cost at least $|F| \geq \sqrt{|F|(|F|-1)} + 0.5$.
\end{lemma}
\begin{proof}
Here, we need to define the new concept of ``\emph{ultra-safe}'' pair.
An ultra-safe pair is a pair of vertex-disjoint edges such that removing it from the graph does not make the resulting graph a star or a bridge graph.
We decompose $F$ in a similar way as we did before. 
However, instead of removing a safe-pair from the graph, we remove an ultra-safe pair in every iteration of the while loop. 
We decompose $F$ using the following procedure.

\begin{figure}[H]
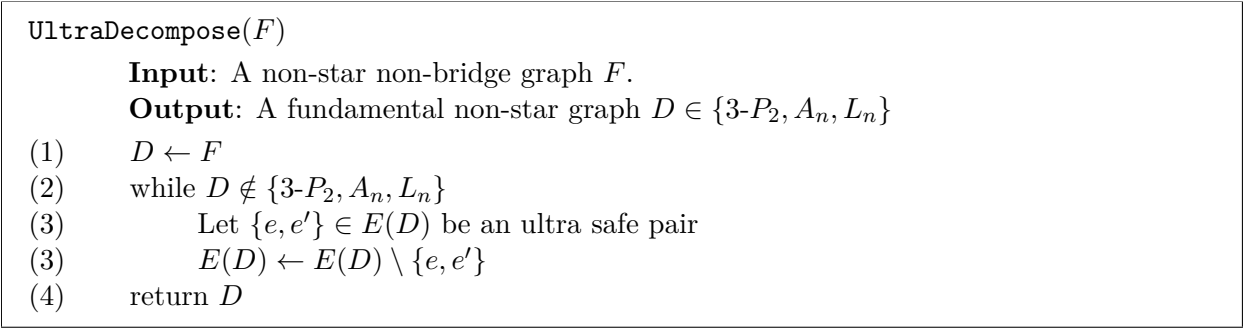

\begin{mdframed}
{\tt UltraDecompose$(F)$} \vspace*{1mm} \\
\hspace*{1.2cm} {\bf Input}: A non-star non-bridge graph $F$.\\
\hspace*{1.2cm} {\bf Output}: A fundamental non-star graph $D \in \{3 \textrm{-} P_{2},A_{n},L_{n}\}$ \vspace*{1mm}\\
\hspace{1cm} (1) \hspace*{0.6cm} $D \gets F$ \\
\hspace{1cm} (2) \hspace*{0.6cm} while $D \notin \{3\textrm{-}P_{2},A_{n},L_{n} \}$ \\
\hspace{1cm} (3) \hspace*{1.5cm} Let $\{e,e'\} \in E(D)$ be an ultra safe pair   \\
\hspace{1cm} (3) \hspace*{1.5cm} $E(D) \gets E(D) \setminus \{e,e'\}$ \\
\hspace{1cm} (4) \hspace*{0.6cm} return $D$
\end{mdframed}
\vspace*{-3mm}
\caption{Decomposition of a non-star non-bridge graph $F$ into fundamental non-star graphs.}
\label{fig:ultra_decompose_procedure}
\end{figure}

First, note that the procedure {\tt UltraDecompose$(F)$} produces a residual graph of type $3$-$P_{2}$ or $A_{n}$. 
It does not produce a residual graph of type $L_{n}$ since we are always removing an ultra-safe pair from the graph, and $L_{n}$ is equivalent to $L_{n,1}$. 
Next, we show that we can always remove an ultra-safe pair from $G$ until we obtain a graph of type $3$-$P_{2}$ or $A_{n}$. 
Consider the $i^{th}$ iteration of the while loop given that it is executed.
Let $D$ be the graph at the start of this iteration. 
It is clear that $D$ is neither a $3$-$P_{2}$ nor $A_{n}$; otherwise, this while loop would not have been executed. 
Also, $D$ can neither be a star nor a bridge graph since an ultra-safe pair was removed in the previous iteration. This fact also holds for the first iteration since the input graph is neither a star nor a bridge graph. It implies that $D$ is a non-star graph but not a fundamental non-star graph. 
Since the graph is not a fundamental non-star graph, it must contain a safe pair. 
Let $e_{1} \equiv (u_{1},v_{1})$ and $e_{2} \equiv (u_{2},v_{2}) $ form a safe pair in $D$. 
If $\{e_{1},e_{2}\}$ is also an ultra-safe pair, we are done.
On the other hand, if $\{e_{1},e_{2}\}$ is not an ultra-safe pair, we show that there is another ultra-safe pair in $D$. 
Since $\{e_{1},e_{2}\}$ is a safe pair but not an ultra safe-pair, removing it would make the resulting graph, a bridge graph $L_{p,q}$. Therefore, $D$ is composed of a graph of type $L_{p,q}$, and the two additional edges: $e_{1}$ and $e_{2}$. 
Let $b \equiv (s_{1},s_{2})$ denote the bridge edge of $L_{p,q}$. 
We split the analysis into two cases, based on the orientation of $e_{1}$ and $e_{2}$ in the graph.
For each case, we show that $D$ contains an ultra-safe pair.
\begin{itemize}
    \item \uline{Case-I}: {\it At least one of the two edges, $e_{1}, e_{2}$ connects a left pendent vertex with a right one.}
    
    Without loss of generality, let $e_{1} \equiv (u_{1},v_{1})$ be the edge that connects a left pendent vertex $u_1$ with a right pendent vertex $v_1$.
    We claim that $e_{1}$ and the bridge edge $b$,  form an ultra-safe pair. 
    Indeed, suppose we remove this pair from the graph. 
    The resulting graph would be a non-star graph since $(s_{1},u_{1})$ and $(s_{2},v_{1})$ are still present in the graph, and they form a $2$-$P_{2}$. 
    Therefore, the pair $\{e_{1},b\}$ satisfies the condition of being a safe pair. 
    Furthermore, the resulting graph is a non-bridge graph since there is no common edge incident on $(s_{1},u_{1})$ and $(s_{2},v_{1})$. 
    This proves that $\{e_{1},b\}$ is an ultra-safe pair.
    
    \item  \uline{Case-II}: {\em Neither $e_{1}$ nor $e_{2}$ connects any left pendent vertex with a right one.} 
    
    First, let us consider the possibility that both the edges $e_{1}$ and $e_{2}$ are incident on the bridge edge, i.e., $e_{1}$ is incident on $s_{1}$ and $e_{2}$ is incident on $s_{2}$. 
    Then the graph $D$ is a bridge graph of the form $L_{p+1,q+1}$ which is not possible as per earlier discussion. 
    Hence, without loss of generality, we can assume that $e_{1}$ is not incident on $b$. 
    Now, we claim that the pair $\{e_{1},b\}$ forms an ultra-safe pair. 
    Note that it forms a $2$-$P_{2}$ since both edges are vertex-disjoint. 
    Now, suppose we remove this pair from the graph. 
    Then in the resulting graph $S_{p}$ and $S_{q}$ are not connected by any edge since $e_{2}$ does not connect left and right pendent vertices. Such a graph can neither be a bridge graph or a star graph, since both of these graphs are connected graphs.
\end{itemize}
The above discussion implies that there $D$ always contains an ultra-safe pair unless $D$ is of type $3$-$P_{2}$ or $A_{n}$.
This means that if the procedure {\tt UltraDecompose}($F$) runs the while loop `$t$' times, then $F$ is composed of $t$ ultra-safe pairs and a fundamental non-star graph $D \in  \{3 \textrm{-} P_{2}, A_{n}\}$.
Based on this decomposition, we bound the optimal cost of $F$ in the following manner:
\begin{align*}
    \Phi^{*}(F) &\geq t \cdot \Phi^{*}(2\textrm{-}P_{2}) + \Phi^{*}(D) &&\\
    &= 2 \cdot t + \Phi^{*}(D), \quad && \textrm{(using Corollary~\ref{corollary:cost_non_star})}\\
    &\geq 2 \cdot t + (|F|-2t), && \textrm{(using Lemma~\ref{lemma:cost_3P2} and~\ref{lemma:cost_An})}\\
    &= |F| &&\\
    &> \sqrt{|F|(|F|-1)}+1/2,  &&\textrm{(using Lemma~\ref{lemma:basic_bound})}
\end{align*}
This completes the proof of the lemma.
\end{proof}

\noindent In the next two subsections, we bound the vertex cover size of any non-star graph $F$ in terms of the extra cost $\delta(F)$.

\subsection{Vertex Cover for Matching Size Two}\label{section:VC_matching_2}
In this section, we show that any graph with a maximum matching of size exactly two has a vertex cover of size at most $2\left( \sqrt{2}+1 \right)\delta(F)+1.62$. Let $C_{5}$ denote a cycle on five vertices. In the following lemma, we show that $C_{5}$ is the only graph with a maximum matching of size two and a vertex cover of size three. The rest of the graphs with a maximum matching of size two, have a vertex cover of size two.

\begin{lemma}\label{lemma:vertex_cover_2matching}
Let $F$ be any graph other than $C_{5}$. If $F$ has a maximum matching of size two, it has a vertex cover of size two. Furthermore, $C_{5}$ has a vertex cover of size three.
\end{lemma}

\begin{proof}
Let $M$ be a maximum matching of $F$. Let $e_{1} \equiv (u_{1},v_{1})$ and $e_{2} \equiv (u_{2},v_{2})$ denote the edges in $M$. 
Let $V_{M}$ denote the vertex set spanned by $M$, i.e., $V_{M} \coloneqq \{u_{1},v_{1},u_{2},v_{2}\}$. 
Let $U$ denote the unmatched edges of $F$. i.e., $U \coloneqq E(F) \setminus M$. 
Note that all edges in $U$ are incident on at least one of the matching edges; otherwise it forms a matching of size three and this contradicts the fact that $F$ has a maximum matching of size two. 
Let $U_{1}$ denote the edges in $U$ that are incident on exactly one of the matching edges and $U_{2} = U \setminus U_{1}$ be the remaining unmatched edges. 
In other words, $U_{2}$ contains the edges that have their both endpoints in $V_{M}$. 

First, we claim that no two edges in $U_{1}$ can be incident on different endpoints of the same matching edge. 
For the sake of contradiction, suppose $(x,u_{1})$ and $(y,v_{1})$ are the edges in $U_{1}$ such that $x,y \notin V_{M}$. 
If $x = y$, it forms a triangle $(x,u_{1},v_{1})$, which is not allowed. 
If $x \neq y$, we have a matching of size three -- $\{(x,u_{1}), (y,v_{1}), (u_{2},v_{2}) \}$. 
This contradicts the fact that $F$ has a maximum matching $M$ of size two. 
Therefore, $U_1$ cannot contain both $(x,u_{1})$ and $(y,v_{1})$.
Similarly, $U_1$ cannot contain both $(x,u_{2})$ and $(y,v_{2})$.
Now, without loss of generality, we can assume that all the edges in $U_{1}$ have their one endpoint in the vertex-set: $\{u_{1},u_{2}\}$. 
Let us divide the remaining analysis into following two cases based on the existence of edge $(v_{1},v_{2})$ in the graph.

\begin{itemize}
    \item \uline{Case 1}: $(v_{1},v_{2}) \notin E(F)$. 
    
$U_{2}$ can only contain the following edges: $(u_{1},v_{2}), (u_{2},v_{1})$, and $(u_{1},u_{2})$. 
Note that all edges in $U_{2}$ have at least one endpoint in $\{u_{1},u_{2}\}$. 
Previously, we showed all edges of $U_{1}$ are incident on $\{u_{1},u_{2}\}$. Based on both of these facts, we can cover all edges in $U$ using only two vertices: $\{u_{1},u_{2}\}$. Furthermore, these vertices also cover the matching edges in $M$. Therefore, all edges of the graph are covered, and we have a vertex cover of size two.

\item
\uline{Case 2}: $(v_{1},v_{2}) \in E(F)$.

Now, note that $U_{2}$ can not contain the edges: $(u_{1},v_{2})$ and $(u_{2},v_{1})$,
since they form the triangles: $(v_{1},v_{2},u_{1})$ and $(v_{1},v_{2},u_{2})$. 
However, $U_{2}$ can contain the edge: $(u_{1},u_{2})$. 
Let us consider the following two sub-cases based on the existence of $(u_{1},u_{2})$ in the graph.

\begin{enumerate}[label=(\alph*)]
    \item \uline{Sub-case:} $(u_{1},u_{2}) \in E(F)$.
    
    We claim that either all the edges in $U_{1}$ are incident on $u_{1}$ or all of them are incident on $u_{2}$. For the sake of contradiction, let $(x,u_{1})$ and $(y,u_{2})$ be the edges in $U_{1}$ such that $x,y \notin V_{M}$. If $x = y$, it forms a triangle $(x,u_{1},u_{2})$, which is not allowed. If $x \neq y$, we get a matching of size three -- $\{(x,u_{1}), (y,u_{2}), (v_{1},v_{2}) \}$ which contradicts the fact that $F$ has a maximum matching $M$ of size two. 
    Without loss of generality, we can assume that all edges in $U_{1}$ are incident on $u_{1}$. 
    Therefore, we can cover all edges of $U_{1}$ using only $u_{1}$. 
    Furthermore, $u_{1}$ covers one of the matching edge $(u_{1},v_{1}) \in M$ and the edge $(u_{1},u_{2}) \in U_{2}$. Only two edges remain uncovered in the graph, which are $(u_{2},v_{2}) \in M$ and $(v_{1},v_{2}) \in U_{2}$. We cover both these edges by picking the vertex $v_{2}$. Thus, we get a vertex cover of size two.

    \item \uline{Sub-case:} $(u_{1},u_{2}) \notin E(F)$.
    
    Let us consider the case when all the edges in $U_{1}$ are incident on either $u_{1}$ or $u_{2}$. 
    In this case, either $\{u_{1},v_{2}\}$ or $\{u_{2},v_{1}\}$ forms a vertex cover of size two. Hence, we are done. Let us consider the other case. Suppose, there are two edges $(x,u_{1})$ and $(y,u_{2})$ in $U_{1}$ such that $x,y \notin V_{M}$. 
    If $x \neq y$, we get a matching of size three --  $\{(x,u_{1}), (y,u_{2}), (v_{1},v_{2})\}$, which is not possible. On the other hand, if $x = y$, then the only possibility is that $F$ is a cycle of length five -- $C_{5} : (x,u_{1},v_{1},v_{2},u_{2})$. In this case, the vertex cover of $F$ is of size $3$.
\end{enumerate}
\end{itemize}
We showed that all graphs with maximum matching $2$ has a vertex cover of size $2$ except for $C_{5}$ that has a vertex cover of size $3$.
This completes the proof of the lemma.
\end{proof}

\begin{corollary}
Let $F$ be any graph with a maximum matching of size two. If the graph is not a $C_{5}$, it has a vertex cover of size at most $ \del +1.62$.
\end{corollary}
\begin{proof}
In Lemma~\ref{lemma:decompose_lemma}, we showed that any graph $F$ has an extra cost at least 0.158. In other words, $\delta(F) \geq 0.158$. It is easy to see that $|VC(F)| = 2 \leq \del +1.62$. Hence proved.
\end{proof}

\noindent Next, we consider the particular case of $C_{5}$. The following lemma bounds the optimal $1$-median cost of $C_{5}$.

\begin{lemma}\label{lemma:cost_5cycle}
The optimal $1$-median cost of $C_{5}$ is at least $\sqrt{|C_5|(|C_5|-1)} + 0.622$.
\end{lemma}
\begin{proof}
Let $C_{5}$ be $(u,v,w,x,y)$. We decompose the graph into two fundamental non-star graphs: $A_{1} \colon \{ (u,v), (w,x) \}$ and $A_{2} \colon \{ (v,w),(x,y),(y,u) \}$. The following sequence of inequalities bound the optimal cost of $C_{5}$.
\begin{align*}
\Phi^{*}(C_{5}) &\geq \Phi^{*}(A_{1}) + \Phi^{*}(A_{2}) &&\\
&\geq 2 + (3 + 0.095) \quad 
&&\textrm{(using Statement 1 and 2, of Lemma~\ref{lemma:cost_An})}\\
&= \sqrt{20} + \left( 5-\sqrt{20} \right) + 0.095 &&\\ 
&= \sqrt{|C_5|(|C_5|-1)} + 0.622 && \textrm{($\because |C_5| = 5$)}
\end{align*}
This completes the proof of the lemma.
\end{proof}

\begin{corollary}
The graph $C_{5}$ has a vertex cover of size at most $\del[C_{5}] + 1.62$.
\end{corollary}
\begin{proof}
Since $\delta(C_{5}) \geq 0.622$, we get 
$|VC(C_{5})| = 3 \leq \del[C_{5}] + 1.62$ which proves the corollary.
\end{proof}

\subsection{Vertex Cover for Matching Size at least Three}\label{section:VC_matching_3}
In this section, we show that any non-star graph $F$, with a maximum matching of size at least three, has a vertex cover of size at most $1.8 + \del$. 
First, let us define some notations. 
Let $M$ denote a maximum matching of $F$, and $G_{M}$ denote the subgraph spanned by $M$. 
Let $F'$ be the graph obtained by removing $M$ from $F$, i.e., $F' = (V, E(F) \setminus M)$. 
Let $L$ denote a maximum matching of $F'$, and $G_{L}$ denote the subgraph spanned by $L$. 
We call $L$ the \emph{second maximum matching} of $F$ after $M$.
Now, we remove $L$ from the graph. 
Let $F''$ be the graph obtained by removing $L$ from $F'$,
i.e., $F'' = (V, E(F) \setminus (M \cup L))$. 
Recall that in this entire discussion, we are using $|.|$ to denote the number of edges in any given graph.

Now, we obtain a relation between the vertex cover size and extra cost of a graph. 
To establish this relation, we show that both of them are proportional to the number of vertex disjoint edges in the graph. 
For example, a graph with a maximum matching $M$ has a vertex cover of size at most $2 |M|$. 
Similarly, a set of $m$ vertex-disjoint edges has an extra cost of $\left( \sqrt{2}-1 \right) \cdot \sqrt{m(m-1)}$ (using Corollary~\ref{corollary:cost_non_star}). 
Also, note that a star graph which has a maximum matching of size one, has an extra cost of only zero. 
In the next two lemmas, we formally establish these relations of the vertex cover size and extra cost in terms of number of vertex-disjoint edges in the graph. 
Then we will use the two lemmas to bound the vertex cover size in terms of the extra cost. 
First, let us bound the vertex cover size in terms of $|L|$ and $|M|$.

\begin{lemma}\label{lemma:general_vertex_cover}
Any non-star graph $F$ has a vertex cover of size at most $(|M| + |L| - 1)$.
\end{lemma}
\begin{proof}
Let $G_{ML}$ denote the graph spanned by the edge set $M \cup L$. 
Note that there are no odd cycles in $G_{ML}$; otherwise there would be two adjacent edges in the cycle that would belong to the same matching set $L$ or $M$. Thus, $G_{ML}$ is a bipartite graph. 
There is a well-known result that says that in a bipartite graph, the size of a maximum matching is equal to the size of a vertex cover~\cite{fermat:book_graph_theory_applications_bondy}. 
Therefore, $G_{ML}$ has a vertex cover of size exactly $|M|$. 
Let $S'$ denote a vertex cover of $G_{ML}$.

Now, we give an incremental construction of a vertex cover of the entire graph $G$. Let this vertex cover be denoted by $S$. 
Initially, we add all vertices of $S'$ to $S$. 
Therefore, at this stage, $S$ covers all edges in $L$ and $M$ which means that for every edge in $L$, at least one of its endpoints must belong to $S$. 
Now, we include its other endpoint in $S$ as well and we do this for all edges in $L$. 
We observe that $S$ now covers all edges in $F'$ since $L$ is a maximum matching of $F'$, and all edges in $F'$ are incident on $L$. 
Therefore, $S$ covers all edges in $G$, and has a size of $|M|+|L|$.

Our main goal is to obtain a vertex cover of size $|M| + |L| - 1$. 
We again give an incremental construction and let $\mathcal{S}$ denote this incrementally constructed vertex cover. 
Initially, $\mathcal{S}$ is empty. 
Let us color the edges of the graph. 
We color the edges in $M$ with red color, $L$ with green color, and $E(F'')$ with blue color. 
Note that non-red edges are the edges of the graph $F'$.
Now, for every edge in $L$ except one, we add both its endpoints in $\mathcal{S}$. 
Let $e' \equiv (u',v') \in L$ be the remaining edge of $L$. Now, we remove all the edges of $F$ covered by $\mathcal{S}$. 
Let the resulting graph be $G_{\mathcal{S}}$. 
$G_{\mathcal{S}}$ contains some red edges, some blue edges, and exactly one green edge $e'$. 
Also, note that all non-red edges in $G_{\mathcal{S}}$ form a star graph. 
This is because, if they form a non-star graph, it would have a matching of size at least two and this matching together with the removed green edges form a matching of $F'$ of size $\geq |L|+1$. This contradicts with the fact that that $F'$ has a maximum matching of size $|L|$. Therefore, non-red edges of $G_{\mathcal{S}}$ form a star graph. 
Now, let us construct a vertex cover of $G_{\mathcal{S}}$. Let $R$ be the set of red edges in $G_{\mathcal{S}}$. 
Let $NR$ be the set of non-red edges in $G_{\mathcal{S}}$. Further, assume that $NR \coloneqq \{(u,v_{1}), (u,v_{1}), \dotsc, (u,v_{t})\}$, i.e., all non-red edges are incident on a common vertex $u$. We consider three different cases depending on the number of red edges in $R$. For each of these cases, we construct a vertex cover for $G_{\mathcal{S}}$.

\begin{enumerate}
    \item \uline{Case 1}: $|R| \leq |M| - |L|$.
    
    Since $NR$ forms a star graph, we cover it using a single vertex $u$. For every edge in $R$, we pick one vertex per edge in the vertex cover. Thus, all edges of $G_{\mathcal{S}}$ are covered. So, the size of the entire vertex cover of $F$ is $(|M| - |L|) + 1 + 2\cdot(|L| - 1) = |M| + |L| - 1$. 
    
    \item \uline{Case 2}: $|R| \geq |M| - |L| + 2$.
    
    In this case, $R$ and the removed green edges form a matching of size $|M| + 1$ which contradicts with the fact that $G$ has the maximum matching of size $|M|$. Therefore, we can rule out this case.
    
    \item \uline{Case 3}: $|R| = |M| - |L| + 1$
    
    Here, we claim that every non-red edge in $NR$ must be incident on some red edge in $R$. 
    For the sake of contradiction, suppose this is not true and there is a non-red $e_{i} \in NR$ that is not incident on any of the red edges in $R$. 
    It is easy to see that $\{e_{i}\} \cup R \cup L \setminus \{ e'\}$ forms a matching of size $|M|+1$. It contradicts that $F$ has a matching of size $|M|$. Therefore, each non-red edge in $NR$ must be incident on some red edge in $R$. Moreover, note that no two edges in $NR$ can be incident on the same red edge $(r_{1},r_{2}) \in R$. Otherwise, it would form a triangle -- $(u,r_{1},r_{2})$, which is not allowed. Now, for every edge in $R$, we pick exactly one of its endpoints. This is the endpoint that it shares with some non-red edge in $NR$ if one exists; otherwise an arbitrary endpoint is picked. Thus, we cover the edges in $G_{\mathcal{S}}$ using only $|R|$ vertices. So, the size of the vertex cover of the entire graph $F$ in this case is $|R| + 2|L| - 2 = |M| + |L| - 1$.
\end{enumerate}
This completes the proof of the lemma. 
\end{proof}

\noindent Now, we bound the extra cost of a graph in terms of $|M|$ and $|L|$. There are some special graph instances for which we do the analysis separately. 
For the following lemma, we assume that $|L| \geq 3$ and $F''$ is a non-star non-bridge graph. 

\begin{lemma}\label{lemma:case0_2}
Let $|L| \geq 3$, and $F''$ be a non-star non-bridge graph. Then the extra cost of $F$ is at least $(\sqrt{2}-1) \cdot \left( \, |M| + |L| \, \right) - 1.06$.
\end{lemma}
\begin{proof}
We decompose $F$ into three subgraphs: $G_{M}$, $G_{L}$, and $F''$. 
It gives the following bound on the optimal cost of $F$.
\begin{align*}
    \Phi^{*}(F) &\geq \Phi^{*}(G_{M}) + \Phi^{*}(G_{L}) + \Phi^{*}(F'')\numberthis \label{eqn:VC_0}
\end{align*}

\noindent We already know the bounds on the optimal costs of $G_{M}$, $G_{L}$ and $F''$. That is,

\begin{itemize}
    \item 
$
\Phi^{*}(G_{M}) \stackrel{\mbox{\tiny{Corollary~\ref{corollary:cost_non_star}}}}{\geq} \sqrt{2} \cdot \st[|M|] \sinq{\textrm{(Lemma~\ref{lemma:basic_bound},$|M| \geq 3$)}}{\geq} \sqrt{2} \cdot \left( |M| - (3 - \sqrt{6}) \right)
\geq \sqrt{2} \cdot |M| -0.78.
$

\item 
$
\Phi^{*}(G_{L}) \sinq{Corollary~\ref{corollary:cost_non_star}}{=} \sqrt{2} \cdot \st[|L|] \sinq{\textrm{(Lemma~\ref{lemma:basic_bound}, $|L| \geq 3$)}}{\geq} \sqrt{2} \cdot \left( |L| - (3 - \sqrt{6}) \right) 
\geq \sqrt{2} \cdot |L| -0.78.
$

\item  $\Phi^{*}(F'') \sinq{Lemma~\ref{lemma:ultra_decompose}}{\geq} |F''|.$
\end{itemize}

\noindent Substituting the above values in equation~\ref{eqn:VC_0}, we obtain the following inequality:
\begin{align*}
    \Phi^{*}(F) &\geq |M| + |L| + |F''| + (\sqrt{2}-1) \cdot \left( |M| + |L|\right) - 1.56&&\\
    &= |F| + (\sqrt{2}-1) \cdot \left( |M| + |L|\right) - 1.56 && \\
    &> \sqrt{|F|(|F|-1)} + 0.5 + (\sqrt{2}-1) \cdot \left( |M| + |L|\right) - 1.56 &&\textrm{(using Lemma~\ref{lemma:basic_bound})}\\
    &= \sqrt{|F|(|F|-1)} + (\sqrt{2}-1) \cdot \left( |M| + |L|\right) - 1.06 &&
\end{align*}
This completes the proof of the lemma.
\end{proof}

\noindent Note that in Lemma~\ref{lemma:general_vertex_cover}, we bound the vertex cover size in terms of $|M|$ and $|L|$. 
Then in Lemma~\ref{lemma:case0_2}, we bound the extra cost in terms of $|M| + |L|$. 
Now, we put these two results together and obtain a relation between the extra cost and vertex cover size.

\begin{corollary}
Let $|L| \geq 3$ and $F''$ is a non-star non-bridge graph. 
Then $F$ has a vertex cover of size at most $1.6 + \del$.
\end{corollary}
\begin{proof}
The proof follows from the following sequence of inequalities:\\
$
     1.6 + \del \sinq{Lemma~\ref{lemma:case0_2}}{\geq} |M| + |L| + 1.6 - \left(\sqrt{2}+1 \right) (1.06) 
     > |M| + |L| - 1
     \sinq{Lemma~\ref{lemma:general_vertex_cover}}{=} |VC(F)|.
$
\end{proof}

\noindent

\noindent There are some special graph instances for which either Lemma~\ref{lemma:general_vertex_cover} gives a weak bound on the vertex cover size or Lemma~\ref{lemma:case0_2} gives a weak bound the extra cost of the graph. This would give an overall weak relation between the vertex cover size and extra cost of the instances. Therefore, we analyse such instances separately. We divide the remaining instances into the following five categories.
\begin{enumerate}
    \item $|L| = 0$: In this case, we show $|VC(F)| \leq (\sqrt{2}+1)\delta(F) + 0.551$.
    \item $|L| = 1$: In this case, we show $|VC(F)| \leq (\sqrt{2}+1)\delta(F) + \mathbf{1.8}$.
    \item $|L| = 2$ and $F'$ is a bridge graph: In this case, we show $|VC(F)| \leq (\sqrt{2}+1)\delta(F) + 1.53$.
    \item $|L| = 2$ and $F'$ is a non-bridge graph: In this case, we show $|VC(F)| \leq (\sqrt{2}+1)\delta(F) + 1.68$.
    \item $|L| \geq 3$ and $F''$ is a bridge graph: In this case, we show $|VC(F)| \leq (\sqrt{2}+1)\delta(F) + 1.4$.
\end{enumerate}

\noindent We analyse these instance one by one. Note that the overall technique remains the same. That is, we first bound the vertex cover size in terms of $|M|$ and $|L|$. Then, we obtain a lower bound on the extra cost in terms of $|M|$ and $|L|$. And, finally we state a corollary (similar to the corollary above) combining these two results. Also, note that for all the following cases we will consider $|M|\geq 3$ since we have already dealt with the case $|M| = 2$ in Section~\ref{section:VC_matching_2}. 

\subsubsection{Case: $\mathbf{|L| = 0}$}

The following lemma is trivial.
\begin{lemma}[Vertex Cover]\label{lemma:case1_1}
If $|L| = 0$, $F$ has a vertex cover of size exactly $|M|$.
\end{lemma}

\begin{lemma}[Extra Cost]\label{lemma:case1_2}
If $|L| = 0$, the extra cost of $F$ is exactly $\left( \sqrt{2}-1 \right) \cdot \st[|M|] $.
\end{lemma}
\begin{proof}
The proof simply follows from Corollary~\ref{corollary:cost_non_star}.
\end{proof}

\begin{corollary}
If $|L| = 0$, $F$ has a vertex of size at most $0.551 + \del$.
\end{corollary}
\begin{proof}
The proof follows from the following series of inequalities:\\
$
    0.551 + \del \sinq{Lemma~\ref{lemma:case1_2}}{=} 0.551 + \st[|M|] 
    \sinq{\textrm{(Lemma~\ref{lemma:basic_bound}, $|M| \geq 3$)}}{\geq} 0.551 + |M| - (3 - \sqrt{6}) > |M|
    \sinq{Lemma~\ref{lemma:case1_1}}{=} |VC(F)|.
$
\end{proof}

\subsubsection{Case: $\mathbf{|L| = 1}$}

Note that the condition $|L| = 1$ is equivalent to $F'$ being a star graph.
\begin{lemma}[Vertex Cover]\label{lemma:case2_1}
If $|L| = 1$, $F$ has a vertex cover of size exactly $|M|$
\end{lemma}
\begin{proof}
The proof follows from Lemma~\ref{lemma:general_vertex_cover} and substituting $|L| = 1$.
\end{proof}

\begin{lemma}[Extra Cost]\label{lemma:case2_2}
If $|L| = 1$, the extra cost of $F$ is at least $ \left(\sqrt{2}-1 \right) \left( |M|\right) - 0.743$
\end{lemma}
\begin{proof}

Let $e$ be some edge in $E(F')$. 
The edge $e$ must incident on some edge of $M$; otherwise $M$ would not be a maximum matching. 
Furthermore, $e$ can only be incident on at most two edges of $M$. 
Let us define the edges $l_{1}$ and $l_{2}$ in the graph depending on the orientation of $e$ in the graph.
\begin{itemize}
    \item If $e$ is incident on two edges of $M$, then $l_{1}$ and $l_{2}$ are defined as the corresponding incident edges in $M$.
    \item If $e$ is incident on only one edge of $M$, then $l_{1} \in M$ is defined as the incident edge and $l_{2}$ is defined as any other edge in $M$. 
\end{itemize}
Let $M' \coloneqq (M \setminus \{\ell_{1},\ell_{2}\} ) \cup \{e\}$ and $L' = E(F') \cup \{\ell_{1},\ell_{2}\} \setminus \{e\}$.
Given this, note that $M'$ forms a matching of size $(|M|-1)$
and $L'$ spans a graph of type $A_{n}$ for $n \geq 1$. 
Let $G_{M'}$ denote the graph spanned by $M'$, and $G_{L'}$ denote the graph spanned by $L'$.
We decompose $F$ into these two subgraphs, i.e., $G_{M'}$ and $G_{L'}$. 
It gives the following bound on the optimal cost of $F$.
\begin{align*}
    \Phi^{*}(F) &\geq \Phi^{*}\left(  G_{M'}  \right) + \Phi^{*}\left(  G_{L'}  \right) \numberthis \label{eqn:VC_2}
\end{align*}

\noindent We already know the bounds on the optimal costs of $G_{M'}$ and $G_{L'}$. 
That is,

\begin{itemize}
    \item 
$
\Phi^{*}(G_{M'}) \sinq{Corollary~\ref{corollary:cost_non_star}}{\geq} \sqrt{2} \cdot \st[|M'|]
\sinq{\textrm{(Lemma~\ref{lemma:basic_bound},  $|M'| \geq 2$)}}{\geq} \sqrt{2} \cdot \left( |M'| - (2 - \sqrt{2}) \right).
$

\item 
$
\Phi^{*}(G_{L'}) \sinq{\textrm{(Lemma~\ref{lemma:cost_An} statement 1)}}{\geq} |L'|.
$
\end{itemize}

\noindent We substitute the above values in Equation (\ref{eqn:VC_2}). This gives the following inequality:
\begin{align*}
    \Phi^{*}(F) &\geq |M'| + |L'| + \left( \sqrt{2}-1 \right) |M'| + 2 - 2 \sqrt{2} \\
    &= |M| + |F'| + \left(\sqrt{2} - 1\right) |M| + 3-3\sqrt{2} &&\\
    &\quad \textrm{(substituting $|M'| = |M|-1$ and $|L'| = |F'| + 1$)} &&\\
    &= |F| + \left(\sqrt{2} - 1\right) |M| + 3-3\sqrt{2} && \left( \because |F| = |M| + |F'| \right)\\
    &> \sqrt{|F|(|F|-1)} + 0.5 + \left(\sqrt{2} - 1\right) |M| + 3-3\sqrt{2} && \textrm{(using Lemma~\ref{lemma:basic_bound})} \\
    &> \sqrt{|F|(|F|-1)} + \left(\sqrt{2} - 1\right) |M| -0.743 &&
\end{align*}
This completes the proof of the lemma.
\end{proof}

\begin{corollary}
If $|L| = 1$, then $F$ has a vertex cover of size at most $1.8 + \del$.
\end{corollary}
\begin{proof}
The proof follows from the following sequence of inequalities:\\
\[
     1.8 + \del \sinq{Lemma~\ref{lemma:case2_2}}{\geq} |M| + 1.8 - \left(\sqrt{2}+1 \right) (0.743) 
     > |M|
     \sinq{Lemma~\ref{lemma:case2_1}}{=} |VC(F)|.
\]
\end{proof}


\subsubsection{Case: $|L|=2$ and $\mathbf{F'}$ is Bridge Graph}

Since $F'$ is a bridge graph, $|L| = 2$. For this case, Lemma~\ref{lemma:general_vertex_cover} gives a vertex cover of size at most $|M| + 1$. However, we show a stronger bound than this in the following lemma.

\begin{lemma}[Vetex Cover] \label{lemma:case3_1}
If $F'$ is a bridge graph $L_{p,q}$ for some $p,q \geq 1$, then $F$ has a vertex cover of size $|M|$.
\end{lemma}
\begin{proof}
Let $b \equiv (u,v)$ be the bridge edge of $L_{p,q}$. 
Suppose $b$ be incident on an edge $e \in M$. 
Without loss of generality, we can assume that $u$ is the common endpoint of $e$ and $b$. 
Let us pick the vertex $u$ in the vertex cover and remove the edges covered by it. 
Let $G'$ denote the resulting graph. 
Further, let $M'$ denote a maximum matching of $G'$.
Now, we claim that $|M'|$ = $|M|-1$. 
For the sake of contradiction, assume that $|M'| \geq |M|$. 
Then the edge $b$ and matching set $M'$ would together form a matching of size $|M|+1$ and this would contradict that $F$ has a maximum matching $M$ of size $|M|$. 
Now, suppose we choose $M' \equiv M \setminus \{e\}$ as the maximum matching of $G'$ and let $L'$ be the second maximum matching after $M'$.
If we remove the edges of $M'$ from $G'$, the remaining graph would be a star graph. 
Therefore, the size of second maximum matching $L'$ is exactly one.
Now, using Lemma~\ref{lemma:general_vertex_cover}, we can cover $G'$ using $|M'| + |L'| - 1$ vertices. 
Thus the vertex cover (including the vertex $u$) of the entire graph $F$ has a size at most $|M'| + |L'| = |M|$. 
This proves the lemma.
\end{proof}

\begin{lemma}[Extra Cost] \label{lemma:case3_2}
If $F'$ is a bridge graph $L_{p,q}$ for some $p,q \geq 1$, then the extra cost of $F$ is at least $ \left(\sqrt{2}-1 \right) \cdot |M| - 0.28$.
\end{lemma}
\begin{proof}
We decompose $F$ into two subgraphs: $G_{M}$ and $F'$. It gives the following bound on the optimal cost of $F$.
\begin{align*}
    \Phi^{*}(F) &\geq \Phi^{*}(G_{M}) + \Phi^{*}(F') \numberthis \label{eqn:VC_3}
\end{align*}

\noindent We already know the bounds on the optimal costs of $G_{M}$ and $F'$. That is,

\begin{itemize}
    \item 
$
\Phi^{*}(G_{M}) \sinq{\textrm{Corollary~\ref{corollary:cost_non_star}}}{\geq} \sqrt{2} \cdot \st[|M|]
\sinq{\textrm{(Lemma~\ref{lemma:basic_bound}, $|M| \geq 3$)}}{\geq} \sqrt{2} \cdot \left( |M| - (3 - \sqrt{6}) \right)
\geq \sqrt{2} \cdot |M| - 0.78.
$

\item 
$
\Phi^{*}(F') \sinq{\textrm{Lemma~\ref{lemma:decompose_lemma}}}{\geq} |F'| - 0.342.
$
\end{itemize}

\noindent We substitute the above values in Equation~(\ref{eqn:VC_3}). It gives the following inequality:
\begin{align*}
    \Phi^{*}(F) &\geq |F'| + |M| + \left( \sqrt{2}-1 \right) |M| -1.122 \\
    &= |F| + \left(\sqrt{2} - 1\right) |M| -1.122  &&\\
    &> \sqrt{|F|(|F|-1)} + 0.5 + \left(\sqrt{2} - 1\right) |M| -1.122 && \textrm{(using Lemma~\ref{lemma:basic_bound})} \\
    &> \sqrt{|F|(|F|-1)} + \left(\sqrt{2} - 1\right) |M| -0.63 &&
\end{align*}
This completes the proof.
\end{proof}

\begin{corollary}
If $F'$ is a bridge graph $L_{p,q}$ for some $p,q \geq 1$, then $F$ has a vertex cover of size at most $1.53 + \del$.
\end{corollary}
\begin{proof}
The proof follows from the following sequence of inequalities:\\
$
     1.53 + \del \sinq{\textrm{Lemma~\ref{lemma:case3_2}}}{\geq} |M| + 1.53 - \left(\sqrt{2}+1 \right) (0.63) > |M| \sinq{\textrm{Lemma~\ref{lemma:case3_1}}}{=} |VC(F)|.
$
\end{proof}

\subsubsection{Case: $|L| = 2$ and $\mathbf{F'}$ is Non-Bridge Graph}

\begin{lemma}[Vertex Cover] \label{lemma:case4_1}
If $|L|=2$ and $F'$ is a non-bridge graph, then $F$ has a vertex cover of size at most $|M|+1$.
\end{lemma}
\begin{proof}
The proof simply follows from Lemma~\ref{lemma:general_vertex_cover} for $|L| = 2$.
\end{proof}

\begin{lemma}[Extra Cost] \label{lemma:case4_2}
If $|L|=2$ and $F'$ is a non-bridge graph, the extra cost of $F$ is at least $ \left(\sqrt{2}-1 \right) \cdot |M| + 0.5$
\end{lemma}
\begin{proof}
We decompose $F$ into two subgraphs: $G_{M}$ and $F'$. It gives the following bound on the optimal cost of $F$.
\begin{align*}
    \Phi^{*}(F) &\geq \Phi^{*}(G_{M}) + \Phi^{*}(F') \numberthis \label{eqn:VC_4}
\end{align*}

\noindent We already know the bounds on the optimal costs of $G_{M}$ and $F'$. That is,

\begin{itemize}
    \item 
$
\Phi^{*}(G_{M}) \sinq{\textrm{ Corollary~\ref{corollary:cost_non_star}}}{\geq} \sqrt{2} \cdot \st[|M|] \sinq{\textrm{(Lemma~\ref{lemma:basic_bound}, $|M| \geq 3$)}}{\geq} \sqrt{2} \cdot \left( |M| - (3 - \sqrt{6}) \right) \geq \sqrt{2} \cdot |M| - 0.78.
$

\item 
$
\Phi^{*}(F') \sinq{\textrm{ Lemma~\ref{lemma:ultra_decompose}}}{\geq} |F'|.
$
\end{itemize}

\noindent We substitute the above values in Equation~(\ref{eqn:VC_4}). It gives the following inequality:
\begin{align*}
    \Phi^{*}(F) &\geq |F'| + |M| + \left( \sqrt{2}-1 \right) |M| - 0.78 \\
    &= |F| + \left(\sqrt{2} - 1\right) |M| - 0.78 && \left( \because |F| = |F'| + |M| \right)\\
    &> \sqrt{|F|(|F|-1)} + 0.5 + \left(\sqrt{2} - 1\right) |M| -0.78 && \textrm{(using Lemma~\ref{lemma:basic_bound})} \\
    &= \sqrt{|F|(|F|-1)} + \left(\sqrt{2} - 1\right) |M| -0.28 &&
\end{align*}
This completes the proof.
\end{proof}

\begin{corollary}
If $F'$ is a non-star non-bridge graph, then $F$ has a vertex cover of size at most $1.68 + \del$.
\end{corollary}
\begin{proof}
The proof follows from the following sequence of inequalities:\\
$
     1.68 + \del \sinq{\textrm{ Lemma~\ref{lemma:case4_2}}}{\geq} |M| + 1.68 - \left(\sqrt{2}+1 \right) (0.28)
     > |M| + 1
     \sinq{\textrm{ Lemma~\ref{lemma:case4_1}}}{=} |VC(F)|.
$
\end{proof}


\subsubsection{Case: $\mathbf{|L|} \geq 3$ and $\mathbf{F''}$ is Bridge Graph}

Since $|L| \geq 3$, Lemma~\ref{lemma:general_vertex_cover} gives a vertex cover of size at most $|M| + |L| -1$, which is at least $|M| + 2$. 
However, we can obtain a stronger bound than this if $F''$ is a bridge graph as shown in the following lemma.

\begin{lemma}\label{lemma:case5_1}
If $|L| \geq 3$ and $F''$ is a bridge graph $L_{p,q}$ for some $p,q \geq 1$, then $F$ has a vertex cover of size at most $|M| + 1$.
\end{lemma}
\begin{proof}
We will incrementally construct a vertex cover $S$ of size $|M| + 1$. 
Initially, $S$ is empty, i.e., $S = \emptyset$. 
Let $b \equiv (u,v)$ be the bridge edge of $L_{p,q}$. 
We will add both vertices $u$ and $v$ to the set $S$, so that it covers all edges in $F''$.
Now, we remove all the edges in the graph that are covered by $u$ and $v$. Let $M'$ and $L'$ be the remaining sets corresponding to $M$ and $L$, respectively. Let $G'$ be the graph spanned by the edge set $M' \cup L'$. 
Now, observe that $G'$ does not contain any odd cycles; otherwise there would be two adjacent edges in the cycle that would belong to the same set $M'$ or $L'$. 
Moreover, $G'$ has a maximum matching of size at most $|M|-1$. 
This is because, the edge $b$ is vertex-disjoint from every edge of $G'$ and if $G'$ has a matching of size at least $|M|$, then this matching together with $b$ form a matching of size $|M|+1$. 
This contradicts the fact that $F$ has the maximum matching of size $|M|$.
Since $G'$ is bipartite and has a matching of size at most $|M|-1$, it admits a vertex cover of size $|M|-1$ (using the \emph{Kőnig's Theorem}~\cite{fermat:book_graph_theory_applications_bondy}). Thus, the vertex cover (including the vertices $u$ and $v$) of the entire graph has a size at most $|M|+1$.
This completes the proof of the lemma.
\end{proof}

\begin{lemma}\label{lemma:case5_2}
If $|L| \geq 3$ and $F''$ is a bridge graph $L_{p,q}$ for some $p,q \geq 1$, then the extra cost of $F$ is at least $(\sqrt{2}-1) \cdot \left( |M| + |L|\right) - 1.41$.
\end{lemma}
\begin{proof}
We decompose $F$ into three subgraphs: $G_{M}$, $G_{L}$, and $F''$. Then, it gives the following bound on the optimal cost of $F$.
\begin{align*}
    \Phi^{*}(F) &\geq \Phi^{*}(G_{M}) + \Phi^{*}(G_{L}) + \Phi^{*}(F'')\numberthis \label{eqn:VC_5}
\end{align*}

\noindent We already know the bounds on the optimal costs of $G_{M}$, $G_{L}$ and $F''$. That is,

\begin{itemize}
    \item 
$
\Phi^{*}(G_{M}) \sinq{\textrm{ Corollary~\ref{corollary:cost_non_star}}}{\geq} \sqrt{2} \cdot \st[|M|]
\sinq{\textrm{(Lemma~\ref{lemma:basic_bound}, $|M| \geq 3$)}}{\geq} \sqrt{2} \cdot \left( |M| - (3 - \sqrt{6}) \right)
\geq \sqrt{2} \cdot |M| -0.78.
$

\item 
$
\Phi^{*}(G_{L}) \sinq{\textrm{(using Corollary~\ref{corollary:cost_non_star})}}{=} \sqrt{2} \cdot \st[|L|]
\sinq{\textrm{(Lemma~\ref{lemma:basic_bound}, $|L| \geq 3$)}}{\geq} \sqrt{2} \cdot \left( |L| - (3 - \sqrt{6}) \right)
\geq \sqrt{2} \cdot |L| -0.78.
$

\item 
$
\Phi^{*}(F'') \sinq{\textrm{ Lemma~\ref{lemma:decompose_lemma}}}{\geq} |F''| - 0.342.
$
\end{itemize}

\noindent We substitute the above values in Equation~(\ref{eqn:VC_5}). It gives the following inequality:
\begin{align*}
    \Phi^{*}(F) &\geq |M| + |L| + |F''| + (\sqrt{2}-1) \cdot \left( |M| + |L|\right) - 1.902&&\\
    &= |F| + (\sqrt{2}-1) \cdot \left( |M| + |L|\right) - 1.902 && \\
    &> \sqrt{|F|(|F|-1)} + 0.5 + (\sqrt{2}-1) \cdot \left( |M| + |L|\right) - 1.902 && \textrm{(using Lemma~\ref{lemma:basic_bound})}\\
    &> \sqrt{|F|(|F|-1)} + (\sqrt{2}-1) \cdot \left( |M| + |L|\right) - 1.402 &&
\end{align*}
This completes the proof.
\end{proof}

\begin{corollary}
If $|L| \geq 3$ and $F''$ is a bridge graph $L_{p,q}$ for some $p,q \geq 1$, then $F$ has a vertex cover of size at most $1.4 + \del$.
\end{corollary}
\begin{proof}
The proof follows from the following sequence of inequalities:\\
$
     1.4 + \del \sinq{\textrm{ Lemma~\ref{lemma:case5_2}}}{\geq} |M| + |L| + 1.4 - \left(\sqrt{2}+1 \right) (1.402)
     \sinq{\textrm{($|L| \geq 3$)}}{>} |M| + 1
     \sinq{\textrm{ Lemma~\ref{lemma:case5_1}}}{=} |VC(F)|. 
$
\end{proof}
\noindent
This completes the analysis for all graph instances.

\section{Bi-criteria Hardness of Approximation}
In the previous section, we showed that the $k$-median problem cannot be approximated to any factor smaller than $(1+\veps)$, where $\veps$ is some positive constant. 
The next step in the {\em beyond worst-case} discussion is to discuss bi-criteria approximation algorithms.
That is, suppose we allow the algorithm to choose more than $k$ centers.
Then does it produce a solution that is close to the optimal solution with respect to $k$ centers?
Since the algorithm is allowed to output more number of centers we can hope to get a better approximate solution.
An interesting question in this regard would be: \emph{Is there a PTAS (polynomial time approximation scheme) for the $k$-median/$k$-means problem when the algorithm is allowed to choose $\beta k$ centers for some constant $\beta>1$?} In other words, is there an $(1+\veps, \beta)$-approximation algorithm? Note that here we compare the cost of $\beta k$ centers with  the optimal cost with respect to $k$ centers. See Section~\ref{section:introduction} for the definition of $(\alpha,\beta)$ bi-criteria approximation algorithms.

In this section, we show that even with $\beta k$ centers,
the $k$-means/$k$-median problems cannot be approximated within any factor smaller than $(1+\veps')$, for some constant $\veps' > 0$.
The following theorem state this result formally.

\begin{theorem}[$k$-median]\label{theorem:bicriteria_kmedian}
For any constant $1 < \beta<1.015$, there exists a constant $\veps > 0$ such that  there is no $(1+\veps,\beta)$-approximation algorithm for the $k$-median problem assuming the Unique Games Conjecture. 
\end{theorem}

\begin{theorem}[$k$-means]\label{theorem:bicriteria_kmeans}
For any constant $1 < \beta<1.28$, there exists a constant $\veps > 0$ such that  there is no $(1+\veps,\beta)$-approximation algorithm for the $k$-means problem assuming the Unique Games Conjecture. Moreover, the same result holds for any $1 < \beta< 1.1$ under the assumption that $\mathsf{P} \neq \mathsf{NP}$. 
\end{theorem}

\noindent
First, let us prove the bi-criteria inapproximability result corresponding to the $k$-median objective. 

\subsection{Bi-criteria Inapproximability: $k$-Median}\label{section:bicriteria_kmedian}

In this subsection, we give a proof of Theorem~\ref{theorem:bicriteria_kmedian}. 
Let us define a few notations. 
Suppose $\mathcal{I} = (\X,k)$ is some $k$-median instance. 
Then, $OPT(\X,k)$ denote the optimal $k$-median cost of $\X$. 
Similarly, $OPT(\X,\beta k)$ denote the optimal $\beta k$-median cost of $\X$ (or the optimal cost of $\X$ with $\beta k$ centers).
We use the same reduction as we used in the previous section for showing the hardness of approximation of the $k$-median problem. 
Based on the reduction, we establish the following theorem.

\begin{theorem}\label{theorem:bicriteria_reduction_kmedian}
There is an efficient reduction from \vc on bounded degree triangle-free graphs $G$ (with $m$ edges) to Euclidean $k$-median instances $\mathcal{I} = (\X,k)$ that satisfies the following properties:
\begin{enumerate}
    \item If $G$ has a vertex cover of size $k$, then $OPT(\X,k) \leq m-k/2$
    \item For any constant $1 < \beta < 1.015$, there exists constants $\veps, \delta > 0$ such that  if $G$ has no vertex cover of size $\leq (2 - \veps)\cdot k$, then  $OPT(\X,\beta k) \geq m-k/2+\delta k$.
\end{enumerate}
\end{theorem}

\begin{proof}
Since the reduction is the same as we discussed in Section~\ref{subsection:kmeans_comparison} and~\ref{section:kmedian_inapproximability}, we keep all notations the same as before. Also, note that Property 1 in this theorem is the same as Property 1 of Theorem~\ref{theorem:reduction_kmedian}. 
Therefore, the proof is also the same as we did in Section~\ref{section:completeness}. 
Now, we directly move to the proof of Property 2.

The proof is almost the same as we gave in Section~\ref{section:soundness}. However, it has some minor differences since we consider the optimal cost with respect to $\beta k$ centers instead of $k$ centers. 
Now, we prove the following contrapositive statement: ``For any constants $1 < \beta < 1.015$ and $\veps>0$, there exists constants $\veps, \delta > 0$ such that  if $OPT(\X,\beta k) < (m-k/2+\delta k)$ then $G$ has a vertex cover of size at most $(2-\veps)k$''. 
Let $\mathcal{C}$ denote an optimal clustering of $\X$ with $\beta k$ centers. 
We classify its optimal clusters into two categories: (1) \textit{star} and (2) \textit{non-star}. Further, we sub-classify the star clusters into the following two sub-categories:
\begin{enumerate}
    \item[(a)] Clusters composed of exactly one edge. Let these clusters be: $P_{1},P_{2},\dotsc,P_{t_{1}}$. 
    \item[(b)] Clusters composed of at least two edges. Let these clusters be: $S_{1},S_{2},\dotsc,S_{t_{2}}$.
\end{enumerate}
    
\noindent Similarly, we sub-classify the non-star clusters into the following two sub-categories:
\begin{enumerate}
    \item[(i)] Clusters with a maximum matching of size two. Let these clusters be: $W_{1}, W_{2}, \dotsc,W_{t_{3}}$ 
    \item[(ii)] Clusters with a maximum matching of size at least three. Let these clusters be: $Y_{1}, Y_{2}, \dotsc,Y_{t_{4}}$  
\end{enumerate}

\noindent Note that $t_{1} + t_{2} + t_{3} + t_{4}$ equals $\beta k$. Suppose, we first compute a vertex cover of all the clusters except the single edge clusters: $P_{1},\dotsc,P_{t_{1}}$. Let that vertex cover be $VC'$.
Now, some vertices in $VC'$ might also cover the edges in $P_{1},\dotsc,P_{t_{1}}$. Suppose there are $t_{1}'$ single edge clusters that remain uncovered by $VC'$. Without loss of generality, we assume that these clusters are $P_{1},\dotsc,P_{t_{1}'}$. By Lemma~\ref{lemma:single_star_vertex_cover}, we can cover these cluster with $(\frac{2t_{1}'}{3} + 8 \delta k) \leq (\frac{2t_{1}}{3} + 8 \delta k)$ vertices; otherwise the graph would have a vertex cover of size at most $(2k - \delta k)$, and the proof of Property 2 would be complete. 
Now, we bound the vertex cover of the entire graph in the following manner.
\begin{align*}
|VC(G)| &\leq  \sum_{i = 1}^{t_{1}}|VC(P_{i})| + \sum_{i = 1}^{t_{2}}|VC(S_{i})| + \sum_{i = 1}^{t_{3}}|VC(W_{i})| + \sum_{i = 1}^{t_{4}}|VC(Y_{i})|\\
&\leq \left( \frac{2t_{1}}{3} + 8\delta k \right) + t_{2} + \sum_{i = 1}^{t_{3}} \left( \del[W_{i}] + 1.62  \right) + \sum_{i = 1}^{t_{4}} \left( \del[Y_{i}] + 1.8  \right), \\
& \hspace{9cm} \textrm{(using Lemmas~\ref{lemma:non_star_vertex_cover_1},~\ref{lemma:non_star_vertex_cover_2}, and~\ref{lemma:single_star_vertex_cover})} \\
&= (0.67)t_{1} + 8\delta k + t_{2} + (1.62)t_{3} + (1.8) t_{4}  + \left( \sqrt{2}+1 \right) \left( \sum_{i = 1}^{t_{3}} \delta(W_{i}) + \sum_{i = 1}^{t_{4}} \delta(Y_{i})
\right)
\end{align*}

\noindent Since the optimal cost $OPT (\X,\beta k) = \mathlarger{\sum}_{j = 1}^{\beta k} \sqrt{m_{j}(m_{j}-1)} + \mathlarger{\sum}_{i=1}^{t_{3}}  \delta(W_{i}) + \mathlarger{\sum}_{i=1}^{t_{4}}  \delta(Y_{i}) \leq m-k/2+\delta k$, we get $\mathlarger{\sum}_{i=1}^{t_{3}}  \delta(W_{i}) + \mathlarger{\sum}_{i=1}^{t_{4}}  \delta(Y_{i}) \leq m-k/2+\delta k - \mathlarger{\sum}_{j = 1}^{\beta k} \sqrt{m_{j}(m_{j}-1)}$. We substitute this value in the previous equation, and get the following inequality:

\begin{align*}
|VC(G)| &\leq  (0.67)t_{1} + 8\delta k + t_{2} + (1.62) t_{3} +(1.8) t_{4} + \left( \sqrt{2}+1 \right) \cdot \left( m-k/2 -\sum_{j=1}^{\beta k} \sqrt{m_{j}(m_{j}-1)} + \delta k \right)
\end{align*}

\noindent Using Lemma~\ref{lemma:basic_bound}, we obtain the following inequalities:
\begin{enumerate}
    \item For $P_{j}$, $\st[m(P_{j})] \geq m(P_{j}) - 1$ since $m(P_{j}) = 1$
    \item For $S_{j}$, $\st[m(S_{j})] \geq m(S_{j}) - (2 - \sqrt{2})$ since $m(S_{j}) \geq 2$
    \item For $W_{j}$, $\st[m(W_{j})] \geq m(W_{j}) - (2 - \sqrt{2})$ since $m(W_{j}) \geq 2$
    \item For $Y_{j}$, $\st[m(Y_{j})] \geq m(Y_{j}) - (3 - \sqrt{6})$ since $m(Y_{j}) \geq 3$
\end{enumerate} 
We substitute these values in the previous equation, and get the following inequality:

\begin{align*}
|VC(G)| &\leq  (0.67)t_{1} + 8\delta k + t_{2} + (1.62) t_{3} +(1.8) t_{4}  + \left( \sqrt{2}+1 \right) \cdot \left( m-k/2  - \sum_{j=1}^{t_{1}} \left( m(P_{j}) - 1\right) + \right.\\
&  \quad \left.  - \sum_{j=1}^{t_{2}} \left( m(S_{j}) - (2-\sqrt{2})\right) -  \sum_{j=1}^{t_{3}} \left( m(W_{j}) - (2-\sqrt{2})\right) - \sum_{j=1}^{t_{4}} \left( m(Y_{j}) - (3-\sqrt{6})\right)  + \delta k \right)
\end{align*}

\noindent Since $m = \mathlarger{\sum}_{j = 1}^{t_{1}} \, \, m(P_{j}) + \mathlarger{\sum}_{j = 1}^{t_{2}} \, \, m(S_{j}) + \mathlarger{\sum}_{j = 1}^{t_{3}} \, \, m(W_{j})  + \mathlarger{\sum}_{j = 1}^{t_{4}} \, \, m(Y_{j}) $, we get the following inequality:

\begin{align*}
|VC(G)| &\leq  (0.67)t_{1} + 8\delta k + t_{2} + (1.62) t_{3} +(1.8) t_{4}  + \left( \sqrt{2}+1 \right) \cdot \Bigg( -k/2 + t_{1} + t_{2} \cdot \left( 2-\sqrt{2} \right) + \Bigg. \\
& \hspace{5cm} \Bigg. + \,\, t_{3} \cdot \left( 2-\sqrt{2} \right) + t_{4} \cdot \left( 3-\sqrt{6} \right) + \delta k \Bigg) \\
&=  (0.67)t_{1} + 8\delta k + t_{2} + (1.62) t_{3} +(1.8) t_{4}  + \left( \sqrt{2}+1 \right) \cdot \Bigg( \frac{(\beta-1)k}{2} - \frac{\beta k}{2} + t_{1} + t_{2} \cdot \left( 2-\sqrt{2} \right) + \Bigg. \\
& \hspace{5cm} \Bigg. + \,\, t_{3} \cdot \left( 2-\sqrt{2} \right) + t_{4} \cdot \left( 3-\sqrt{6} \right) + \delta k \Bigg)
\end{align*}
\noindent Now, we substitute $\beta k = t_{1} + t_{2} + t_{3} + t_{4}$, and obtain the following inequality:
\begin{align*}
|VC(G)| &\leq (0.67)t_{1} + 8\delta k + t_{2} + (1.62) t_{3} +(1.8) t_{4}  + \left( \sqrt{2}+1 \right) \cdot \left(  \frac{(\beta-1)k}{2} + \frac{t_{1}}{2} + \frac{t_{2}}{10} + \frac{t_{3}}{10} + \frac{3t_{4}}{50}  +  \delta k \right) \\ 
&=  (1.88)t_{1} + (1.25)t_{2} + (1.87) t_{3} + (1.95) t_{4} + \left(\sqrt{2}+1\right) \cdot \frac{(\beta-1)k}{2} + \left( \sqrt{2}+9 \right) \delta k \\
&< (1.95) \beta k + \left(\sqrt{2}+1\right) \cdot \frac{(\beta-1)k}{2} + \left( \sqrt{2}+9 \right) \delta k \hspace{3cm} \textrm{(using $t_{3} + t_{4} + t_{1} + t_{2} = \beta k $)}\\
&< (3.16) \beta k - (1.21) k + \left( \sqrt{2}+9 \right) \delta k\\
&\leq (2-\veps)k, \quad \textrm{for $\beta < 1.015$ and appropriately small constants $\veps,\delta>0$}
\end{align*}

\noindent This proves Property 2 and it completes the proof of Theorem~\ref{theorem:bicriteria_reduction_kmedian}.
\end{proof} The following corollary states the main bi-criteria inapproximability result for the $k$-median problem.

\begin{corollary}
There exists a constant $\veps' > 0$ such that for any constant $1 < \beta<1.015$, there is no $(1+\veps',\beta)$-approximation algorithm for the $k$-median problem assuming the Unique Games Conjecture.
\end{corollary}
\begin{proof}
In the proof of Corollary~\ref{corollary:reduction_kmedian}, we showed that $k \geq \frac{m}{2\Delta}$ for all the hard \vc instances. 
Therefore, the second property of Theorem~\ref{theorem:bicriteria_reduction_kmedian}, implies that $OPT(\X,\beta k) \geq (m - \frac{k}{2}) + \delta k \geq (1+ \frac{\delta}{2\Delta})\cdot (m - \frac{k}{2})$. 
Thus, the $k$-median problem can not be approximated within any factor smaller than $1 + \frac{\delta}{2\Delta} = 1 + \Omega(\veps)$, with $\beta k$ centers for any $\beta < 1.015$.
\end{proof}

\noindent Now, we prove the bi-criteria inapproximability result corresponding to the $k$-means objective.

\subsection{Bi-criteria Inapproximability: $k$-means}\label{section:bicriteria_kmeans}

Here, we again use the same reduction that we used earlier for the $k$-median problem in Sections~\ref{subsection:kmeans_comparison},~\ref{section:kmedian_inapproximability}, and~\ref{section:bicriteria_kmedian}.
Using this, we establish the following theorem.

\begin{theorem}\label{theorem:bicriteria_reduction_kmeans}
There is an efficient reduction from \vc on bounded degree triangle-free graphs $G$ (with $m$ edges) to Euclidean $k$-means instances $\mathcal{I} = (\X,k)$ that satisfies the following properties:
\begin{enumerate}
    \item If $G$ has a vertex cover of size $k$, then $OPT(\X,k) \leq m-k$
    \item For any $1 < \lambda \leq 2$ and $\beta < \fdfrac{2}{7} \cdot \left( \lambda + \fdfrac{5}{2} \right)$, there exists constants $\veps, \delta > 0$ such that if $G$ has no vertex cover of size $\leq (\lambda - \veps)\cdot k$, then  $OPT(\X,\beta k) \geq m-k+\delta k$. 
\end{enumerate}
\end{theorem}

\noindent This theorem is simply an extension of the result of Awasthi {\it et al.}~\cite{hardness:acks15} to the bi-criteria setting. Now, let us prove this theorem.

\subsubsection{Completeness}
Note that the proof of completeness is already given in~\cite{hardness:acks15}.
Therefore, we just describe the main components of the proof for the sake of clarity. 
To understand the proof, let us define some notations used in~\cite{hardness:acks15}. 
Suppose $F$ is a subgraph of $G$. 
For a vertex $v \in V(F)$, let $d_{F}(v)$ denote the number of edges in $F$ that are incident on $v$. 
Note that, the optimal center for $1$-means problem is simply the centroid of the point set. 
Therefore, we can compute the optimal $1$-means cost of $F$.
The following lemma states the optimal $1$-means cost of $F$.

\begin{lemma}[Claim 4.3~\cite{hardness:acks15}]
Let $F$ be a subgraph of $G$ with $r$ edges. Then, the optimal $1$-means cost of $F$ is $\sum_{v} d_{F}(v) \left( 1- \frac{d_{F}(v)}{r} \right)$
\end{lemma}
 
\noindent The following corollary bounds the optimal $1$-means cost of a star cluster. This corollary is implicitly stated in the proof of Claim 4.4 of \cite{hardness:acks15}.
\begin{corollary}\label{corollary:star_cost_kmedian}
The optimal $1$-means cost of a star cluster with $r$ edges is $r-1$.
\end{corollary}

\noindent Using the above corollary, we give the proof of completeness.
Let $V = \{v_{1},\dotsc,v_{k}\}$ be a vertex cover of $G$.  Let $S_{i}$ denote the set of edges covered by $v_{i}$. If an edge is covered by two vertices $i$ and $j$, then we arbitrarily keep the edge either in $S_{i}$ or $S_{j}$. Let $m_{i}$ denote the number of edges in $S_{i}$. We define $\{\X(S_{1}),\dotsc,\X(S_{k})\}$ as a clustering of the point set $\X$. Now, we show that the cost of this clustering is at most $m-k$. 
Note that each $S_{i}$ forms a star graph with its edges sharing the common vertex $v_{i}$. The following sequence of inequalities bound the optimal $k$-means cost of $\X$.
\[
OPT(\X,k) \leq \sum_{i = 1}^{k} \Phi^{*}(S_{i})  \stackrel{\tiny{(Corollary~\ref{corollary:star_cost_kmedian})}}{=} \sum_{i = 1}^{k} \left( m(S_{i}) -1 \right) = m - k.
\]

\subsubsection{Soundness}
For the proof of soundness, we prove the following contrapositive statement: ``For any constant $1 < \lambda \leq 2$ and $\beta < \frac{2}{7} \cdot \left( \lambda + \frac{5}{2} \right)$, there exists constants $\veps, \delta > 0$ such that if $OPT(\beta k) \leq (m-k+\delta k)$ then $G$ has a vertex cover of size at most $(\lambda-\veps)k$, for $\veps = \Omega(\delta)$.''
Let $\mathcal{C}$ denote an optimal clustering of $\X$ with $\beta k$ centers. 
We classify its optimal clusters into two categories: (1) \textit{star} and (2) \textit{non-star}. 
Suppose there are $t_{1}$ star clusters: $S_{1},\dotsc,S_{t_{1}}$, and $t_{2}$ non-star clusters: $F_{1},F_{2},\dotsc,F_{t_{2}}$. Note that $t_{1} + t_{2}$ equals $\beta k$. The following lemma bounds the optimal $1$-means cost of a non-star cluster.

\begin{lemma}[Lemma 4.8~\cite{hardness:acks15}]\label{lemma:cost_non_star_kmeans}
The optimal $1$-means cost of any non-star cluster $F$ with $m$ edges is at least $m-1+\delta(F)$, where $\delta(F) \geq \frac{2}{3}$. Furthermore, there is an edge $(u,v) \in E(F)$ such that $d_{F}(u) + d_{F}(v) \geq m + 1 -\delta(F)$.
\end{lemma}
In the actual statement of the lemma in \cite{hardness:acks15}, the authors mentioned a weak bound of $\delta(F) > 1/2$. 
However, in the proof of their lemma they have shown $\delta(F) > 2/3 > 1/2$. 
This difference does not matter when we consider inapproximability of the $k$-means problem. 
However, this difference improves the $\beta$ value in bi-criteria inapproximability of the $k$-means problem.

\begin{corollary}[\cite{hardness:acks15}]\label{corollary:vertex_cover_non_star_kmeans}
Any non-star cluster $F$ has a vertex cover of size at most $1 +\fdfrac{5}{2} \cdot \delta(F)$.
\end{corollary}
\begin{proof}
Suppose $(u,v)$ be an edge in $F$ that satisfies the property: $d_{F}(u) + d_{F}(v) \geq m + 1 -\delta(F)$, by Lemma~\ref{lemma:cost_non_star_kmeans}. This means that $u$ and $v$ covers at least $m(F) - \delta(F)$ edges of $F$. We pick $u$ and $v$ in the vertex cover, and for the remaining $\delta(F)$ edges we pick one vertex per edge. Therefore, $F$ has a vertex cover of size at most $2 + \delta(F)$. Since $\delta(F) \geq \fdfrac{2}{3}$, by Lemma~\ref{lemma:cost_non_star_kmeans}, we get $2 + \delta(F) \leq 1 + \frac{5}{2} \cdot \delta(F)$. Hence, $F$ has a vertex cover of size at most $1 + \frac{5}{2} \cdot \delta(F)$. This proves the corollary.
\end{proof}

\noindent Now, the following sequence of inequalities bound the vertex cover size of the enire graph $G$.
\begin{align*}
    |VC(G)| &\leq \sum_{i = 1}^{t_{1}} |VC(S_{i})| + \sum_{i = 1}^{t_{2}} |VC(F_{i})| \\
    &\leq t_{1} + \sum_{i = 1}^{t_{2}} \left( 1 + \frac{5}{2} \cdot \delta(F_{i})\right) \quad \textrm{(using Corollary~\ref{corollary:vertex_cover_non_star_kmeans})}\\
    &= t_{1} + t_{2} + \frac{5}{2} \cdot \sum_{i = 1}^{t_{2}} \delta(F_{i})
\end{align*}

\noindent Since the optimal $k$-means cost $OPT(\X,\beta k) = \mathlarger{\sum}_{i = 1}^{t_{1}} \left( m(S_{i}) - 1 \right) + \mathlarger{\sum}_{i = 1}^{t_{2}} \left( m(F_{i}) - 1 + \delta (F_{i})\right) \leq m - k + \delta k$, and $t_{1}+t_{2} = \beta k$. Therefore, $\mathlarger{\sum}_{i = 1}^{t_{2}}\, \delta(F_{i}) \leq (\beta-1)k + \delta k$. On substituting this value in the previous equation, we get the following inequality:

\begin{align*}
    |VC(G)| &\leq t_{1} + t_{2} + \frac{5}{2} \cdot (\beta - 1)k + \frac{5}{2} \cdot \delta k \\
    &= \beta k + \frac{5}{2} \cdot (\beta - 1)k + \frac{5}{2} \cdot \delta k, \quad \textrm{($\because t_{1} + t_{2} = \beta k$)}\\
    &\leq (\lambda-\veps)k, \quad \textrm{for $\beta < \fdfrac{2}{7} \cdot \left( \lambda + \fdfrac{5}{2} \right)$ and appropriately small constants $\veps, \delta > 0$}
\end{align*}
This proves the soundness condition and it completes the proof of Theorem~\ref{theorem:bicriteria_reduction_kmeans}. Based on this theorem, the following corollary states the main bi-criteria inapproximability result for the $k$-means problem.

\begin{corollary}
For any constant $1 < \beta<1.28$, there exists a constant $\veps' > 0$ such that there is no $(1+\veps',\beta)$-approximation algorithm for the $k$-means problem assuming the Unique Games Conjecture. Moreover, the same result holds for any $1 < \beta < 1.1$ under the assumption that $\mathsf{P} \neq \mathsf{NP}$.
\end{corollary}
\begin{proof}
Suppose \vc can not be approximated to any factor smaller than $\lambda - \veps$, for some constant $\veps, \lambda >0$.
In the proof of Corollary~\ref{corollary:reduction_kmedian}, we showed that $k \geq \frac{m}{2\Delta}$ for all the hard \vc instances. 
In that case, the second property of Theorem~\ref{theorem:bicriteria_reduction_kmeans} implies that $OPT(\X,\beta k) \geq (m -k) + \delta k \geq (1+ \frac{\delta}{2\Delta})\cdot (m - k)$. 
Thus, the $k$-means problem can not be approximated within any factor smaller than $1 + \frac{\delta}{2\Delta} = 1 + \Omega(\veps)$, with $\beta k$ centers. Now, let us compute the value of $\beta$ based on the value of $\lambda$. We know that $\beta < \fdfrac{2}{7} \cdot \left( \lambda + \fdfrac{5}{2} \right)$. Consider the following two cases:
\begin{itemize}
    \item By Theorem~\ref{theorem:VC_UGC}, \vc is hard to approximate within any factor smaller than $2-\veps$ on bounded degree triangle-free graphs assuming the Unique Games Conjecture. Hence $\lambda = 2$ and thus $\beta  < 1.28$ assuming the Unique Game Conjecture.
    \item By Theorem~\ref{theorem:VC_PNP}, \vc is hard to approximate within any factor smaller than $1.36$ on bounded degree triangle-free graphs assuming $\mathsf{P} \neq \mathsf{NP}$. Hence $\lambda = 1.36$ and thus $\beta  < 1.1$ assuming $\mathsf{P} \neq \mathsf{NP}$.
\end{itemize}
This completes the proof of the corollary.
\end{proof}

\section{Conclusion}
\label{section:conclusion}

We showed that the Euclidean $k$-median problem cannot be approximated to any factor smaller than $(1+\veps)$ for some constant $\veps>0$ assuming UGC.
In addition, we gave the bi-criteria hardness of approximation results for the Euclidean $k$-median and $k$-means problems.
Besides trying to improve the inapproximability bounds, one interesting future direction is to check if the Euclidean $k$-means/$k$-median problems are hard to approximate in the bi-criteria setting with $2k$ or more centers. 
It would also be interesting to try designing a $(1+\veps,\beta)$-approximation algorithm for $k$-means and $k$-median, for arbitrary small constant $\veps>0$ and some constant $\beta>1$ that is independent of $\veps$.

\bibliographystyle{alpha}
\bibliography{references}
\end{document}